\documentclass[runningheads]{llncs}
\usepackage{graphicx,graphics}
\usepackage{amsfonts}
\usepackage{latexsym,amssymb,amsmath}
\usepackage{url}
\usepackage{algorithm}
\usepackage{algpseudocode}

\usepackage{color}

\newtheorem{fact}{Fact}

\newcommand{\anglebr}[1]{\langle #1 \rangle}

\newcommand{\cA}{{\cal{A}}}
\newcommand{\cS}{{\cal{S}}}
\newcommand{\MST}{{\mbox{\it MST}}}
\newcommand{\MH}{{\mbox{\it MH}}}
\newcommand{\OPT}{{\mbox{\it OPT}}}

\newcommand{\ReduceFastest}{\mbox{Reduce-Fastest}}
\newcommand{\ReduceMax}{\mbox{Reduce-Max}}

\usepackage[T1]{fontenc}

\def\qed{\hfill$\Box$}
\newcommand{\ep}{\qed}

\newcommand{\omitthis}[1]{}

\begin{document}

\pagestyle{headings}

\title{Perpetual maintenance of machines with different urgency requirements%
\thanks{%
L.~G\k asieniec and J.~Min's work was partially supported by Network Sciences and Technologies (NeST)
at University of Liverpool.
This study has been carried out in the frame of ``the Investments for the future'' 
Programme IdEx Bordeaux -- SysNum (ANR-10-IDEX-03-02). 
R.~Klasing’s research was partially supported by the ANR project TEMPOGRAL (ANR-22-CE48-0001).
C.~Levcopoulos and A.~Lingas work were supported by the Swedish Research Council grants 621-2017-03750 and 2018-04001.
T.~Radzik's work was supported in part by
EPSRC grant EP/M005038/1, ``Randomized algorithms for computer networks.''
Part of this work was done while T.~Radzik was visiting the LaBRI as a guest professor of the University of Bordeaux.
T.~Jurdzinski's work was supported by the Polish
National Science Centre project no.\ 2020/39/B/ST6/03288.
Preliminary versions of some of the results 
presented in this paper appeared in Proc.~43rd International Conference on Current Trends in Theory 
and Practice of Computer Science (SOFSEM 2017), LNCS 10139, pp.~229--240, 
Springer (2017)~\cite{DBLP:conf/sofsem/GasieniecKLLMR17}.
} 
}

\author{
Leszek~G\k asieniec \inst{1} \and 
Tomasz Jurdzi\'nski \inst{2} \and
Ralf~Klasing \inst{3} \and Christos~Levcopoulos \inst{4}\\ Andrzej~Lingas \inst{4} \and Jie~Min \inst{1} \and Tomasz~Radzik \inst{5}}
\institute{{Department of Computer Science}, {University of Liverpool}, {Liverpool, UK},
\email{\{l.a.gasieniec, J.Min2\}@liverpool.ac.uk}.
\and {Institute of Computer Science}, 
{University of Wroc\l{}aw}, {Poland},
\email{tju@cs.uni.wroc.pl}.
\and {CNRS}, {LaBRI}, {Universit\'e de Bordeaux}, {France},
\email{ralf.klasing@labri.fr}.
\and {Department of Computer Science}, {Lund University}, {Lund, Sweden},
\email{\{christos.levcopoulos, andrzej.lingas\}@cs.lth.se}.
\and {Department of Informatics}, {King's College London}, {London, UK}, 
\email{tomasz.radzik@kcl.ac.uk}.
}

\titlerunning{Perpetual maintenance schedules}

\authorrunning{L. G\k asieniec et al.}

\sloppy
\maketitle

\begin{abstract}
A garden $G$ is populated by $n\ge 1$ bamboos $b_1, b_2, ..., b_n$ with the respective daily growth
rates
$h_1 \ge  h_2 \ge \dots \ge h_n$. 
It is assumed that the initial heights of bamboos are zero.
The robotic gardener maintaining the garden regularly
attends bamboos and trims them to height zero
according to some schedule. The {\sl Bamboo Garden Trimming Problem} (BGT) is to
design a perpetual schedule of cuts to maintain the elevation of the bamboo garden as low as possible.
The bamboo garden is a metaphor for a collection of machines which have to be serviced,
with different frequencies, by a robot which can service only one machine at a time.
The objective is to design a perpetual schedule of servicing which minimizes
the maximum (weighted) waiting time for servicing.

\vspace{0.1cm}

We consider two variants of BGT.
In {\sl discrete} BGT the robot trims only one bamboo at the end of each day.
In {\sl continuous} BGT
the bamboos can be cut at any time, however, the robot needs time to move from one bamboo to the next.

\vspace{0.1cm}

For discrete BGT, 
we show 
tighter approximation algorithms for the case when the growth rates are balanced and for the general case.
The former algorithm settles one of the conjectures about the Pinwheel problem.
The general approximation
algorithm improves on the previous best
approximation ratio.
For continuous BGT, we propose approximation algorithms
which achieve 
approximation ratios $O(\log \lceil h_1/h_n\rceil)$ and $O(\log n)$.

\vspace{0.1cm}

\noindent
{\bf Key Words:} Bamboo Garden Trimming problem, BGT problem, Perpetual scheduling, Periodic maintenance, Pinwheel scheduling, Approximation algorithms, Patrolling
\end{abstract}


%
%
\section{Introduction}\label{sec:Intro}
We consider a perpetual scheduling problem in which $n\ge 2$ 
(possibly virtual)
machines need to be attended (serviced) with {\sl known} but possibly different frequencies, i.e.\
some machines need to be attended more often than others. 
The frequencies of attending individual machines 
are specified as positive weights $h_1, h_2, \dots, h_n$ 
and the objective is to design a perpetual schedule of attending the machines 
which minimizes the maximum
weighted time any individual machine waits for the next service.
Since a higher weight $h_i$ (comparing with other weights) means that machine $i$ should be attended relatively more frequently, we refer to the weights also as urgency factors.
The same optimization problem arises when a data stream keeps filling a collection of $n$ 
buffers according to a known distribution: buffer $i$ receives $h_i$ units of data 
in each unit of time. 
The objective is to design a perpetual schedule of emptying the buffers 
which minimizes the maximum
occupancy of any individual buffer.

We model such perpetual scheduling problems using the following metaphor of 
the {\sl Bamboo Garden Trimming (BGT) Problem}.
A collection (garden) of $n\ge 2$ bamboos $b_1, b_2, \ldots, b_n$ 
with known respective daily growth rates $h_1, h_2, \dots, h_n$.
We assume that these growth rates are already arranged into a non-increasing sequence:
$h_1 \ge  h_2 \ge \dots \ge h_n > 0$.
Initially the height of each bamboo is set to zero. 
The robotic gardener maintaining
the garden trims bamboos to height zero according to some schedule.
The height of a bamboo $b_i$ at time 
$t \ge 0$ is equal to $(t - t') h_i$, where $t'$ is the last time when
this bamboo was trimmed, or $t' = 0$, if it has never been trimmed by time~$t$.
The main task of the BGT problem 
is to design a perpetual schedule of cuts to keep the highest bamboo in the garden as low as possible, while complying with some specified constraints on the timing of cutting. 
The basic constraints considered in this paper 
are that the gardener can cut only one (arbitrary) bamboo at the end of each day and
is not allowed to attend the garden at any other times.
Once the gardener has decided which bamboo to trim in the current round (at the end of the current day),  then 
the action of actual trimming is instantaneous. 

Referring back to the two scheduling problems mentioned earlier, 
the heights of the growing bamboos would represent the weighted times the machines 
wait for the next service, or the current occupancy of the data buffers.
The action of cutting a bamboo $b_i$ at the end of the current day represents 
attending machine $i$ or emptying buffer $i$ in the current time slot.
Other problems which can be modeled by BGT 
include the perpetual testing of virtual machines in cloud systems~\cite{AKG15}.
In such systems frequency in which virtual machines are tested for undesirable symptoms vary depending on the importance of dedicated cloud
operational mechanisms.

We consider two variants of the BGT problem. 
The constraint that only one bamboo is cut at the end of each day (round) defines
{\em discrete\/} BGT. The gardener has equal access to all bamboos, so can select in each round any bamboo
for cutting.
In the second variant, {\em continuous\/} BGT, 
we assume that for any two bamboos $b_i$ and $b_j$,
we know the time $t_{i,j} > 0$ (which may be fractional) 
that the robot needs to relocate from $b_i$ to $b_j$.
In this variant 
the time when the next bamboo is trimmed 
depends on how far that bamboo is from the bamboo which has just been trimmed.
As in discrete BGT, when the robot arrives at the bamboo to trim it, 
the actual action of trimming is instantaneous.
In this paper we consider symmetric travel times
(that is, $t_{i,j} = t_{j,i}$) and assume that 
the robot travels always along the fastest route, so the travel times 
satisfy the triangle inequality.
We also assume that the robot is initially at the location of $b_1$.
Discrete BGT is the special case of continuous BGT when all travel times $t_{i,j}$, for $i\neq j$, are the same,
while metric TSP is the special case of continuous BGT when all growth rates $h_i$ are the same.
%

In both discrete and continuous cases, we consider algorithms $\cA$ 
which for an input instance $I$ of the form $(h_i:\: 1\le i \le n)$ 
in the discrete case and $[( h_i:\: 1\le i \le n)$, $( t_{i,j}:\: 1\le i,j\le n)]$
in the continuous case,
produce a perpetual (trimming) schedule $\cA(I)$
as a sequence of indices of bamboos $(i_1, i_2, \ldots)$ which defines the order  
in which the bamboos are trimmed.
We are mainly interested in the {\em approximation ratios\/} of such algorithms, 
which are defined in the usual way.
For an input instance $I$ and a trimming schedule $\cS$ for~$I$, 
let $\MH(\cS)$ denote the supremum of the heights of bamboos over all times $t \ge 0$ when the trimming 
proceeds according to schedule $\cS$, and
let $\OPT(I)$ denote 
the infimum of $\MH(\cS)$ over all schedules $\cS$ for~$I$.
The approximation ratio of 
a schedule $\cS$ is defined as $\MH(\cS) / \OPT(I)$ and
the approximation ratio of an algorithm 
$\cA$ is the supremum of $\MH(\cA(I)) / \OPT(I)$ over all input instances $I$.

Regarding the time complexity of BGT algorithms,
we aim at polynomial preprocessing time followed by computation of the consecutive indices of 
the schedule in poly-logarithmic time per one index. 
We will call algorithms with such performance simply polynomial-time (BGT) algorithms.
The computational complexity of continuous BGT is related to the complexity of TSP, as the latter is a special case of the
former. To discuss computational complexity of discrete BGT, we introduce first a lower bound on the height 
of schedules.


For each instance $I$ of discrete BGT with the sum of the growth rates $H = H(I) =h_1+h_2+\dots+h_n$,
a simple and natural lower bound 
on the maximum height of a bamboo in any schedule 
is $\OPT(I) \ge H$.
Indeed, while the heights of all bamboos are at most $H' < H$,
then during each day the {\sl total height} of all bamboos,
that is, the sum of the current heights of all bamboos,
increases at least by $H - H' > 0$ (the total growth over all bamboos is $H$ but only one
bamboo, of height at most $H'$, is cut).
Thus on some day within the first $\lfloor n H'/(H - H') \rfloor + 1$ days 
the total height of the bamboos must exceed $n H'$, so the height of one 
of the bamboos must exceed $H'$.
Observe also that it cannot happen that the maximum height of a bamboo approaches $H$
but never reaches $H$, because there are only finitely many possible heights of bamboos 
which are less than $H$.

There are instances with $\OPT(I) = H$. The obvious one is the uniform instance $h_i \equiv H/n$. A non-uniform example is the input 
instance $I = (1/2, 1/4, 1/4)$, where all bamboos are kept within the $H= 1$ height 
by the schedule with period $(b_1, b_2, b_1, b_3)$.
An example of an input instance with $\OPT(I) > H$ is 
$I = (7/15, 1/3, 1/5)$, for which $H = 1$ 
but $\OPT(I) = 4/3$.
For this instance, 
the schedule with period $(b_1,b_2,b_1,b_2,b_1,b_3)$ does not let any bamboo grow 
above the height $4/3$.
On the other hand, a schedule which keeps the heights of $b_1$ and $b_2$ 
strictly lower than $4/3$ must cut $b_1$ every other day, implying that $b_2$ must also be cut every other day (after the initial couple of days).
Thus, after the initial couple of days, there are no further days available to cut $b_3$, so its height 
grows to infinity.
If we have only two bamboos and their growth
rates are $h_1=1-\varepsilon$, and $h_2=\varepsilon$, for any $0 < \varepsilon 
\le1/2$, then $\OPT(I) = 2(1-\varepsilon)$, so
can be arbitrarily close to $2$.
We note here that for any instance $I$, 
$\OPT(I) \leq 2 H(I)$ 
(this is explained in Section~\ref{s:pinwheel}).

\subsection*{The context and previous related research}

\setcounter{footnote}{0}

Our paper focuses on perpetual maintenance of a given environment where each vital element has its own, 
possibly unique urgency factor. 
This makes it related to {\em periodic scheduling}~\cite{serafini_mathematical_1989}, 
a series of papers 
on the {\em Pinwheel\/} problems~\cite{chan_general_1992,chan_schedulers_1993,holte_pinwheel:_1989}
including the {\em periodic Pinwheel\/} problem~\cite{holte_pinwheel_1992,lin_pinwheel_1997} and
the {\em Pinwheel scheduling\/} problem~\cite{romer_algorithm_1997},
as well as the concept of {\sl P-fairness} in sharing multiple copies of 
some resource among various tasks~\cite{baruah_proportionate_1996,baruah_pfair_1998}.

The Pinwheel problem introduced in~\cite{holte_pinwheel:_1989} can be viewed as a special case of discrete BGT.
The complexity results for the Pinwheel problem presented in~\cite{holte_pinwheel:_1989} imply that 
for a given $K \ge H$, if there is a schedule with height at most $K$, then there is a cyclic schedule with height 
at most $K$, but the shortest such schedule can have exponential length.\footnote{%
Exponential in the size of the input, assuming that the growth rates are rational numbers given as pairs of integers.
} 
This implies that the decision version of discrete BGT can be solved by considering all cyclic schedules 
of up to exponential length, and this can be implemented in PSPACE. 
Further from~\cite{holte_pinwheel:_1989}, 
while the cyclic schedules of height $H$ can also have exponential length, they have concise polynomial-size
representations, and there is a polynomial-time algorithm for checking if a given concise representation of 
a cyclic schedule of height $H$ is valid. This implies that the restricted decision version of discrete BGT 
which asks if there is a schedule of height $H$ is in NP. 
Jacobs and Longo~\cite{DBLP:journals/corr/JacobsL14}
show that there is no pseudopolynomial
time algorithm solving the Pinwheel problem unless SAT has an exact algorithm running in
expected time $n^{O(\log n \log\log n)}$
and consider the complexity of related problems.
However, the exact complexity of 
the Pinwheel problem remains a long-standing open question.

In related research on minimizing the maximum occupancy of a buffer in a system of $n$ buffers,
the usual setting is a game between the player and the adversary
\cite{bender-etal-2015,cinderella,chrobak-etal-ICALP2001}.
The adversary decides how the fixed total increase of data in each round is distributed among the buffers and tries
to maximize the maximum occupancy of a buffer. 
The player decides which 
buffer (or buffers, depending on the variant of the problem)
should be emptied next and tries to minimize the maximum buffer size. 
The upper bounds developed in this more general context can be translated into upper bounds for our
BGT problems, but our aim is to derive tighter bounds for the case 
when the rates of growth of the occupancy of buffers, or the rates of growth of bamboos in our terminology, 
are fixed and known.
Similar models, under the name of 
``cup (emptying) games'',
have been considered, with recent papers including
\cite{%
DBLP:conf/stoc/BenderFK19,%
DBLP:conf/soda/BenderK21,%
DBLP:conf/soda/Kuszmaul20,%
DBLP:conf/stoc/Kuszmaul21,%
DBLP:conf/icalp/KuszmaulN22%
}.
%

The continuous BGT problem is a natural extension of several classical algorithmic problems with
the focus on {\em monitoring} and {\em mobility}, including the {\em Art Gallery Problem}~\cite{ntafos_gallery_1986} 
and its dynamic extension 
 called the {\em $k$-Watchmen Problem}~\cite{urrutia_art_2000}.
In a more recent work on {\em fence patrolling}~\cite{collins_optimal_2013,czyzowicz_boundary_2011,KaKo15} 
the studies
focus on monitoring vital (possibly disconnected) parts of a linear environment where each point is expected to 
be attended with the same frequency. 
Czyzowicz~{\it et~al.\/}~\cite{CGK+15} study monitoring linear environments by robots prone to faults.
Problems similar to continuous BGT are considered also by 
Baller~{\it et~al.\/}~\cite{DBLP:journals/networks/BallerEHS20},
who focus on special cases 
(special metric spaces) to
investigate the boundary between easy (that is, polynomial) and hard cases,
and by Bosman~{\it et al.\/}
\cite{DBLP:journals/algorithmica/BosmanEJMRS22}, who 
minimize the travel cost subject to the feasibility requirement of maintaining 
the specified minimum frequencies of visiting bamboos. 


Probably the most natural strategy to keep the elevation of the bamboo garden low is the greedy approach
of always moving next to the currently highest bamboo and cutting it. 
This approach, called {\sl \ReduceMax}, is particularly appealing in the context of discrete BGT, where 
there are no travel times to be accounted for.
$\ReduceMax$ was considered recently in the context of 
periodic testing of virtual machines in cloud systems~\cite{AKG15}, and 
was also studied in the adversarial setting of the buffer minimization problems mentioned above.
The results presented in~\cite{cinderella} imply a tight upper bound of $H\cdot(H_{n-1} + 1) = \Theta(H \log n)$
on $\MH(\cS)$ for schedules $\cS$ produced by 
$\ReduceMax$ for a variant of the discrete BGT with
the adversary which in each round arbitrarily 
distributes the total daily growth of $H$ among the bamboos.
Here 
$H_k = \sum_{i=1}^{k}\frac{1}{k} = \Theta(\log k)$ is the $k$-th harmonic number.
While this $O(H \log n)$ upper bound applies obviously 
also to 
our non-adversarial discrete BGT, 
when the growth rates are fixed,
it was a long standing open question 
whether there were instances for which $\ReduceMax$ 
lets some bamboos grow to heights $\Omega(H \log n)$, or
even to heights $\omega(H)$.
The experimental work presented in~\cite{AKG15} pointed towards a conjecture that $\ReduceMax$
keeps the maximum bamboo height within $O(H)$, and 
the question has been finally recently answered in Bil{\`{o}} {\it et al.\/}~\cite{BGLPS22},
where 
a bound of $9H$ on the maximum bamboo height under the 
$\ReduceMax$ algorithm was proven. 
Kuszmaul~\cite{DBLP:conf/spaa/Kuszmaul22}
has recently shown that the maximum height of a bamboo under the $\ReduceMax$ cutting strategy is at most $4H$.

As mentioned above, there are instances for which the optimal maximum height can be arbitrarily close to $2H$
(see also Bil{\`{o}} {\it et al.\/}~\cite{BGLPS22} and Kuszmaul~\cite{DBLP:conf/spaa/Kuszmaul22}).
This, however, does not imply a lower bound greater than $1$ on the
approximation ratio of any algorithm (which is defined with respect to the optimum rather than the lower bound $H$).
We note that the input instance $I = (3/8-\varepsilon, 1/4, 1/4)$, where $0 < \varepsilon < 1/24$ can be arbitrarily small, 
shows that the approximation ratio of $\ReduceMax$ cannot be less than~$9/8$.
For this instance, $\OPT(I) = 1$
(for $\varepsilon < 1/24$) with the optimal schedule repeating 
$(b_1,b_2,b_1,b_3)$,
but the $\ReduceMax$ schedule is
$(b_1,b_2,b_3,b_1, \ldots)$, with 
the first bamboo reaching the height 
$9/8 - 3\varepsilon$.
In Section~\ref{sec:Approx-simpleStartegy} we show
that the approximation ratio of $\ReduceMax$ is not less
than $12/7$.



In~\cite{DBLP:conf/sofsem/GasieniecKLLMR17}, 
which included preliminary versions of some of the results presented in this paper,
we introduced a modification of $\ReduceMax$, which we 
called $\ReduceFastest$, to
show the first simple greedy algorithm
achieving constant approximation ratio.
$\ReduceFastest(x)$, where $x>0$ is 
a parameter of the algorithm, works in the following way.
Keep track of the ``tall'' bamboos, defined as having the 
current height at least $x\cdot H$, and 
cut in each step the tall bamboo with the highest growth
rate (no cutting, if there is no tall bamboo).
In~\cite{DBLP:conf/sofsem/GasieniecKLLMR17},
we presented a detailed proof that the approximation ratio
of $\ReduceFastest(2)$ is at most~$4$.
D'Emidio {\it et al.\/}~\cite{DBLP:journals/algorithms/DEmidioSN19}
conducted an extensive
experimental evaluation of various BGT strategies,
which led them to conjecture that $\ReduceMax$, $\ReduceFastest(2)$, and $\ReduceFastest(1)$
keep the maximum height of a bamboo within $2H$, $3H$, and $2H$, respectively.
Subsequently, 
Bil{\`{o}} {\it et al.\/}~\cite{BGLPS22} proved that for $x = 1 + 1/\sqrt{5} \approx 1.45$,
the maximum bamboo height under the $\ReduceFastest(x)$ strategy is not greater than  
an approximation ratio of $(3+\sqrt{5})/2 \approx 2.62$.
They also presented efficient 
implementations of $\ReduceFastest$ and $\ReduceMax$.

Kuszmaul~\cite{DBLP:conf/spaa/Kuszmaul22}
has recently shown that for any $x \ge 2$, 
the strategy $\ReduceFastest(x)$ keeps all bamboos strictly below $(x+1)\cdot H$, proving this way 
the conjecture from~\cite{DBLP:journals/algorithms/DEmidioSN19}
that $\ReduceFastest(2)$ keeps the bamboos below $3H$.
For a lower bound, Kuszmaul~\cite{DBLP:conf/spaa/Kuszmaul22} has shown
that whatever the value of parameter $x$ is,
$\ReduceFastest(x)$ does not give a bound 
better than $(2.01)\cdot H$, 
disproving the conjecture from~\cite{DBLP:journals/algorithms/DEmidioSN19} that 
$\ReduceFastest(1)$ keeps bamboos below $2H$.
The lower bound of $(2.01)\cdot H$
together with
the bound $\OPT(I) \le 2H(I)$ imply that 
the approximation ratio of $\ReduceFastest$ is at least $2.01/2=1.005$.
In Section~\ref{sec:Approx-simpleStartegy}, we give examples showing better 
lower bounds on the approximation ratio of $\ReduceFastest$.
Bil{\`{o}} {\it et al.\/}~\cite{BGLPS22} and
Kuszmaul~\cite{DBLP:conf/spaa/Kuszmaul22}
consider algorithms which guarantee the $2H$ bound on the maximum 
bamboo height (the best possible bound with respect to $H$), the former focusing 
on resource-efficient algorithms and the 
latter on simplicity of computation.
These algorithms therefore have approximation ratios
at most $2$, but it is not known if those ratios 
are strictly less than $2$.

We refer informally 
to BGT algorithms like $\ReduceMax$ and $\ReduceFastest$ as 
{\em online scheduling}. These algorithms are based on 
simple greedy strategies, the trimming schedule is 
revealed while the cutting progresses, and 
the whole cutting process would naturally
adapt to changing growth rates.
An alternative {\em offline 
scheduling} pre-computes the whole (cyclic) schedule. 
This approach would sacrifice 
the flexibility offered by simple greedy 
strategies but hopefully would give better 
approximation ratios.
Indeed, using the Pinwheel results given in~\cite{holte_pinwheel:_1989}, one can
easily obtain an offline 
algorithm $\cA$ which guarantees $\MH(\cA(I)) \le 2 H(I)$ for each input $I$, so 
a $2$-approximation algorithm
for discrete BGT. 
An efficient implementation 
of this algorithm was developed
in~\cite{BGLPS22}.
Recently, \cite{Ee21} has shown
an offline scheduling algorithm, 
also based on Pinwheel results, which
guarantees $12/7$-approximation --
the best approximation ratio for 
the general problem shown prior 
to our paper.
We emphasise that our online/offline 
categorisation of scheduling algorithms
is informal, and does not refer to the 
availability of input 
(in both cases the whole input is known in advance) but only indicates the
general nature of algorithms.
Similarly, to distinguish algorithms 
like $\ReduceMax$ and $\ReduceFastest$ from
more complex approaches,
Kuszmaul~\cite{DBLP:conf/spaa/Kuszmaul22}
refers informally 
to this type of algorithms as 
\emph{simple algorithms}.

The optimisation objective in our work, and in the related work discussed above, 
is {\it minimizing the maximum\/} (weighted) waiting time (the maximum height of a bamboo). 
Anily~{\it et al.}~\cite{DBLP:journals/dam/AnilyGH98,DBLP:journals/anor/AnilyGH99} 
consider perpetual scheduling of servicing machines with the objective 
of {\it minimizing the average\/} (weighted) time of waiting for maintenance
(using different but equivalent terminology of 'linearly increasing operational costs').
They show that there is always an optimal schedule which is cyclic and propose and evaluate 
various strategies of computing schedules 
for the general case of $n$ machines~\cite{DBLP:journals/dam/AnilyGH98} and
for the special case of $3$ machines~\cite{DBLP:journals/anor/AnilyGH99}.  


\subsection*{Structure of the paper and our contributions}

In Sections~\ref{sec:Approx-simpleStartegy}
to~\ref{sec:impoved-approx}, we consider 
the discrete BGT.
In Section~\ref{sec:Approx-simpleStartegy}, we 
derive some lower bounds on 
the approximation ratios of the simple strategies $\ReduceMax$ and $\ReduceFastest$.
In Section~\ref{sec:BGT-Pinwheel}, we
elaborate on the Pinwheel problem and its relation 
to discrete BGT, laying foundations for our main results.
In Section~\ref{sec:DiscreteBGT-offline}, we present our main approximation algorithm 
for discrete BGT, 
which is an offline algorithm derived by 
further exploration of 
the relation between discrete BGT and
the Pinwheel problem and has approximation ratio
$(1+O(\sqrt{h_1/H}))$. 
The benefits of the relation between the discrete BGT and Pinwheel problems extend both ways. 
On the one hand, our approximation algorithm uses 
properties of the Pinwheel problem. On the other hand, 
the approximation ratio which we achieve settles one of
the conjectures about the Pinwheel problem as explained in Section~\ref{sec:BGT-Pinwheel}.

In Section~\ref{sec:impoved-approx}, 
we turn our attention to general approximation bounds.
As mentioned earlier, the best previous general bound
is $12/7$ given in~\cite{Ee21}. 
We show an algorithm with approximation ratio $8/5+o(1) < 12/7$.
This improvement is based on a new approach of splitting the growth rates into 
two groups of large and small rates,
computing good schedules separately for each group, and then merging
these two schedules into one schedule for all rates.
The algorithm uses the
$(1+O(\sqrt{h_1/H}))$-approximation algorithm to compute a schedule 
for the group of small rates.

In Section~\ref{sec:ContinuousBGT}, we show algorithms for continuous BGT
with approximation ratios $O(\log\lceil{h_1}/{h_n}\rceil)$ and $O(\log n)$.\footnote{%
As Metric TSP is a subproblem of continuous BGT, it is NP-hard to approximate
continuous BGT with a factor better than 123/122 \cite{KLS15}.
}
We also discuss how tight these approximation ratios are. 
We show instances of
continuous BGT such that for any schedule the maximum
bamboo height is greater by a $\Theta(\log n)$
factor than the lower bounds which we use in the analysis of approximation ratios. 
Thus for these input instances our $O(\log n)$-approximation
algorithm computes in fact schedules with constant approximation ratios.
We also show instances for which this algorithm computes $\Theta(\log n)$-approximate
schedules.

\section{Approximation ratio of simple strategies}
\label{sec:Approx-simpleStartegy}

In Section~\ref{sec:Intro}, we presented the previous 
work on the simple strategies $\ReduceMax$ and 
$\ReduceFastest$, which focused on analysing 
the maximum height of any bamboo in relation to 
the sum of the growth rates $H$. 
The value $H$ is, however, only a lower bound on 
the optimum, so the bounds in relation to $H$ do not
necessarily give good bounds on the approximation 
ratios, which refer to the optimal values.
Recall that there are BGT instances for which the 
optimal (minimal) height is arbitrarily close to $2H$.
In this section, we show some lower bounds 
on the approximation ratios of these two strategies.






\paragraph{Approximation ratio of \ReduceMax.} $\;$


We show that the approximation ratio of $\ReduceMax$
is not less than $12/7$ by considering the following
BGT instances.
Let $i=7k+3$, for any integer $k\ge 1$. We have a sequence of $i+1 = 7k+4$ growth rates
partitioned into two groups.
In group 1, we have only the largest 
growth rate 
$h_1 = \frac{3k}{i} = 3/7-\frac{9}{7i}$\/,
while the remaining $i$ growth rates belong to group 2 and are all equal to $\frac{1}{2i}\/.$
It is easy to see that the optimal (minimal) 
height for these instances is at most $1$:
schedule $b_1$ every other time slot and the other
bamboos in the remaining slots in the round-robin
manner.

We now look at the computation of $\ReduceMax$ on this sequence of growths and consider three initial stages defined as follows.
\begin{itemize}
    \item Stage 1: all bamboos in group 2 are smaller than $h_1,$ so $b_1$ is cut during each round.
    \item Stage 2: all bamboos in group 2 are smaller than $2h_1$, but some of them are taller than $h_1$, implying that $b_1$ is cut in every other round,
    whenever it reaches height $2h_1$.
    \item Stage 3: all bamboos in group 2 are smaller than $3h_1$, but some of them are taller than $2h_1$, implying that $b_1$ is cut in every third round, whenever it reaches height $3h_1$. 
\end{itemize}

The main idea is to show not only that Stage 3 is not empty, but also that at the end of Stage 3, there are still 3 bamboos in group 2 which have never been cut. Such bamboos are taller than $3h_1$, so $b_1$ is not cut for 3 consecutive rounds, growing
to the height $4h_1$.

In Stage 1 no bamboos from group 2 are cut, and they all grow to height $h_1,$ at the end of Stage 2, some bamboos from group 2 grow to height $2h_1$, and at the end of Stage 3 some bamboos from group 2 grow to height $3h_1.$ 
As in Stage~2 and Stage 3 some bamboos from group 2 are cut, we need to show that we have enough bamboos in this group for some of them not to be cut in either of the three stages.

For each Stage $j$, $j = 1,2,3$, there is a bamboo in group 2 which adds $h_1$ to its height (growing from $(j-1)h_1$ to $jh_1$), so there are ${h_1}/{\frac{1}{2i}} = h_1\cdot 2i=6k$ rounds in every stage. Thus $3k$ bamboos from group 2 are cut in Stage 2 (one group-2 bamboo cut in every other round in this stage), and $4k$ bamboos from group 2 are cut in Stage 3 (two of them are cut in each sequence of 3 consecutive rounds in this stage). Thus, after all three stages, at most $7k$ different bamboos from group 2 will be cut altogether. Hence, at least 3 bamboos from group 2 will not be cut in either of the three stages. 3 of these bamboos will be cut in the next 3 rounds. Therefore, bamboo $b_1$ will grow to height $4h_1=12/7-o(1)$.




While $\ReduceMax$ pushes 
the maximum bamboo height close to $12/7$
for our BGT instances,
this may be happening only in the initial period of the schedule.
However, it is easy to see that in these instances 
$b_1$ will have to keep growing to the height
of $3h_1$, even in the long run. Otherwise, after some round $t$,
the height of each bamboo would never go
above $2h_1 + o(1) < H$, which is not sustainable.
Thus, denoting by $m_t$ the maximum height of any bamboo at any round 
after round $t$, our input instances show that we can get $m_t \ge 9/7-\varepsilon$,
for arbitrarily small $\varepsilon > 0$.


\paragraph{Approximation ratio of \ReduceFastest.} $\;$

For $0 < x < 1$, the input instance $I=(x,\varepsilon)$,
where $0 < \varepsilon < \min\{x,1-x\}$, shows that the approximation ratio of $\ReduceFastest(x)$ is unbounded, as
noted in Kuszmaul~\cite{DBLP:conf/spaa/Kuszmaul22}.
For this instance, $\OPT(I) = 2x$ with the optimal schedule repeating $(b_1,b_2)$, but the $\ReduceFastest(x)$ schedule never cuts bamboo $b_2$, hence the height of bamboo $b_2$ grows to infinity.

For $1 \le x < 2$, the input instance $I=(x/2-\varepsilon, \varepsilon)$, where $0 < \varepsilon < x/2$, shows that the approximation ratio of $\ReduceFastest(x)$ cannot be less than $3/2$.
For this instance, $\OPT(I) = x-2\varepsilon$ with the optimal schedule repeating $(b_1,b_2)$, but the $\ReduceFastest(x)$ schedule only cuts bamboo $b_1$ every 3 rounds, hence bamboo $b_1$ reaches the height $3x/2 - 3\varepsilon$.

For the parameter $x\geq 2$, 
the input instance $I=(1-\varepsilon,\varepsilon)$, where $0 < \varepsilon < 1/2$, shows that the approximation ratio of $\ReduceFastest(x)$ cannot be less than $3/2$.
For this instance, $\OPT(I) = 2-2\varepsilon$ with the optimal schedule repeating $(b_1,b_2)$, but the $\ReduceFastest(x)$ schedule only cuts bamboo $b_1$ at most every 3 rounds, hence bamboo $b_1$ reaches at least the height $3 - 3\varepsilon$.


\section{Discrete BGT problem and Pinwheel}
\label{sec:BGT-Pinwheel}

%
The input of the Pinwheel problem is a sequence 
$V= \anglebr{f_{1},f_{2},\ldots,f_{n}}$ of integers $2 \le f_{1} \le f_{2}\le \ldots \le f_{n}$
called ({\em pinwheel\/}) {\em frequencies}.
The objective is to specify an infinite sequence $S$ of indices 
drawn from the set ${1,2,\ldots,n}$ such that
for each index $i$, any sub-sequence of $f_{i}$ 
consecutive elements in $S$
includes at least one index $i$, 
or to establish that such a sequence does not exists. 
A sequence $S$ with this property is called a schedule of $V$, and if such a sequence does not exist, then
we say that the sequence of frequencies $V$
is not feasible or that it cannot be scheduled.
The Pinwheel problem is a special case of  discrete BGT: 
an input instance $\anglebr{f_{1},f_{2},\ldots,f_{n}}$ of Pinwheel is feasible if, and only if,
the optimal (minimum) height for the input instance 
$(1/f_{1}, 1/f_{2},\ldots, 1/f_{n})$ of discrete BGT is at most $1$.

The Pinwheel problem was introduced in~\cite{holte_pinwheel:_1989}, 
where some complexity 
results were presented and some classes of feasible sequences of frequencies were established.
It is easy to see
that it is not possible to schedule any instance $V$ whose {\em density} 
$D(V)\equiv \sum_{i=1}^{n}{1}/{f_{i}}$ is greater than $1$,
since in any feasible schedule each frequency $f_i$ takes 
at least ${1}/{f_{i}}$ fraction of the slots.\footnote{%
More precisely, in each prefix of length $T-f_n$ of a feasible schedule, where $T$ is
arbitrarily large,
each frequency $f_i$ must take at least $(1/f_i)(T-f_i)$ slots, implying 
$\sum_{i=1}^{n}{1}/{f_{i}} \le 1$.
}
(This upper bound of~$1$ on the density of a feasible instance of the Pinwheel problem 
is a special case of the lower bound of $H$ on the maximum height for
an instance of discrete BGT.) 
The sequence of frequencies $(2,4,4)$ is an example of a feasible 
instance of the Pinwheel problem with density~$1$.
On the other hand,
the input instances $(2,3,M)$, where $M$ is an arbitrarily large integer,
show that there are instances with densities arbitrarily close to $5/6$ which are not feasible.
One of the conjectures for the Pinwheel problem, which remains open, 
is that $5/6$ is the universal threshold guaranteeing feasibility of input instances.
That is, it has been conjectured (ever since the pinwheel problem was introduced) that
any instance with density at most $5/6$
can be scheduled.
The current best proven bound is $3/4$~\cite{fishburn_pinwheel}.

Our work is related to another conjecture for the Pinwheel problem, made by 
Chan and Chin~\cite{chan_general_1992},
that when the first frequency $f_1$ keeps increasing, then the density threshold guaranteeing 
feasibility keeps increasing to $1$.
To be more precise, 
let ${\cal{V}}$, ${\cal{V}}_{yes}\subseteq {\cal{V}}$ and ${\cal{V}}(f,\Delta)\subseteq {\cal{V}}$
denote the set of all instances $V$
of the Pinwheel problem with density $D(V) \le 1$, the set of all feasible instances and 
the set of instances with $f_1 = f$ and density $D(V) \le \Delta$, respectively. 
Define
$d(f) \equiv \sup\{ \Delta:\; {\cal{V}}(f,\Delta) \subseteq {\cal{V}}_{yes}\}$ as
the density threshold guaranteeing feasibility of instances with the 
first frequency equal to $f$. 
That is, each input instance $(f_1=f, f_2, f_3, \ldots)$ with density 
less than $d(f)$ is feasible, while for each $\epsilon >0$, there 
is an infeasible instance with density less than $d(f)+\epsilon$.
Chan and Chin~\cite{chan_general_1992,chan_schedulers_1993} conjecture 
that $\lim_{f\rightarrow\infty} d(f) = 1$, consider a number 
of heuristics for the Pinwheel problem (referred to as  {\em schedulers})
and analyze  the guarantee density threshold
$$d_{\cal{A}}(f) \equiv \sup\{ D:\; 
\mbox{heuristic ${\cal{A}}$ schedules each instance in  ${\cal{V}}(f,D)$}\}$$
for each considered heuristic~${\cal{A}}$.
They derive lower bounds $\ell_{\cal{A}}(f)$ on the values $d_{\cal{A}}(f)$, but 
for each of their lower bounds, $\lim_{f\rightarrow\infty} \ell_{\cal{A}}(f)$ is strictly less than $1$,
which leaves possibility that $\lim_{f\rightarrow\infty} d_{\cal{A}}(f)$ 
is also strictly less than~$1$.\footnote{%
Chan and Chin~\cite{chan_general_1992,chan_schedulers_1993} 
showed that for some algorithms which they considered $\lim_{f\rightarrow\infty} d_{\cal{A}}(f)$
is actually strictly less than $1$. For the other algorithms, they left unanswered the question whether $\lim_{f\rightarrow\infty} d_{\cal{A}}(f) = 1$.}
Thus~\cite{chan_general_1992,chan_schedulers_1993} left open the question of designing an algorithm ${\cal{A}}$ for which 
$\lim_{f\rightarrow\infty} d_{\cal{A}}(f) = 1$, and there have not been other results in this direction prior to our work.
Such an algorithm 
would immediately imply 
that $\lim_{f\rightarrow\infty} d(f) = 1$.

Our $(1+O(\sqrt{h_1/H}))$-approximate polynomial-time algorithm for discrete BGT
applied to input instances $(1/f_{1}, 1/f_{2},\ldots, 1/f_{n})$
is a polynomial-time scheduler for Pinwheel input instances $\anglebr{f_{1},f_{2},\ldots,f_{n}}$
with the guarantee density threshold $1-O(\sqrt{1/f_1})$.
This threshold tends to $1$ with increasing $f_1$, proving the conjecture that 
$\lim_{f\rightarrow\infty} d(f) = 1$.
To see that this Pinwheel scheduler has indeed 
the guarantee density threshold $1-O(\sqrt{1/f_1})$,
let $c>0$ be a constant such that the approximation ratio of our algorithm
for discrete BGT is at most $1+ c\sqrt{h_1/H}$.
If the density $D(V)$ of a Pinwheel instance $V = \anglebr{f_{1},f_{2},\ldots,f_{n}}$ is at most $1 - c/\sqrt{f_1}$, 
then the BGT schedule computed for the input $(1/f_{1}, 1/f_{2},\ldots, 1/f_{n})$
does not let the height of any bamboo
go above  (note that here $H= D(V)<1$):
\begin{eqnarray*}
D(V)\left((1+ c\sqrt{(1/f_1)/D(V)}\right) & = & D(V) + c\sqrt{D(V)/f_1} \\
& \le & 1 - c/\sqrt{f_1} + c\sqrt{D(V)/f_1} \; \le  \: 1.
\end{eqnarray*}
Thus the computed schedule is a feasible schedule for the Pinwheel instance.

%

\section{Discrete BGT by offline scheduling}\label{sec:DiscreteBGT-offline}

In this section we focus on offline scheduling which permits tighter approximation
than the approximation of online algorithms
discussed 
in Section~\ref{sec:Approx-simpleStartegy}.
These results are achieved by exploring 
the relationship between BGT and the Pinwheel scheduling problem.
Some known facts about Pinwheel scheduling 
give immediately a 2-approximation BGT algorithm. Our main result in this section is 
a $(1+O(\sqrt{h_1/H}))$-approximation algorithm. 
One of the consequences of this approximation algorithm is that it settles the conjecture made for the Pinwheel problem
that if the first (smallest) frequency keeps increasing, the 
density threshold which guarantees feasibility increases to $1$ (cf.~Section~\ref{sec:BGT-Pinwheel}).



\subsection{Reducing discrete BGT to Pinwheel scheduling}\label{s:pinwheel}

We use the notation for the Pinwheel problem introduced in Section~\ref{sec:BGT-Pinwheel}.
Each feasible input sequence of frequencies $f_1 \le f_2 \le \cdots \le f_n$ has
the density $D = \sum_{i=1}^{n} 1/f_i$ at most one.
We also know that any instance with density at most $3/4$ is feasible~\cite{fishburn_pinwheel}, 
and it is conjectured that any instance with density at most $5/6$ is feasible.
%
%
To see the relationship between the BGT and the Pinwheel
problem, define for a BGT input instance $I = (h_1 \ge  h_2 \ge \dots \ge h_n > 0)$
the sequence of frequencies
$f'_{i}=H/h_{i}$, $i = 1,2, \ldots, n$.
This sequence is a \emph{pseudo-instance} $\anglebr{f'_1, f'_2, \ldots, f'_n}$
of Pinwheel (pseudo, since these frequencies
are rational numbers rather than integers) with  density:
\[
D'={\displaystyle \sum_{i=1}^{n}}\frac{1}{f'_{i}}
={\displaystyle \sum_{i=1}^{n}}\frac{h_{i}}{H}=1.
\]
We multiply the frequencies $f'_i$ by $1+\delta$ to obtain another 
pseudo-instance $f''_{i}= f'_{i}(1+\delta)$, $i = 1,2, \ldots, n$, 
with the density reduced to $1/(1+\delta) <1$,
where $\delta>0$ is a suitable parameter.
%
%
Finally, we obtain a (proper) instance $V(I,\delta) = \anglebr{f_1, f_2, \ldots, f_n}$ of the Pinwheel problem
by reducing each frequency $f''_i$ to an integer $f_i \le f''_i$.
We require that the integer frequencies~$f_i$ are not greater than the rational frequencies $f''_i$, but
we do not specify at this point their exact values, 
leaving this to concrete algorithms.
Reducing frequencies $f''_i$ to $f_i$ increases the density of the sequence.
The room for this increase of the density was made by 
the initial decrease of the density to $1/(1+\delta)$.


\begin{lemma}\label{l:delta}
If $I$ is an instance of BGT, $\delta > 0$ and an instance
$V(I,\delta)$ of the Pinwheel problem is feasible, then a feasible schedule for
this Pinwheel instance $V(I,\delta)$ is
a $(1+\delta)$-approximation schedule for the BGT instance $I$.
\end{lemma}
\begin{proof}
In a feasible schedule for $V(I,\delta)$, two consecutive occurrences of an index $i$
are at most $f_i \le {H(1+\delta)}/{h_i}$ slots apart.
This means that if this schedule is used for the BGT instance $I$, then the height of $b_i$ 
is never greater than 
$h_i\cdot f_i \le H(1+\delta)$.
\ep\end{proof}

In view of Lemma~\ref{l:delta}, 
the goal is to get a {\em feasible\/} instance $V(I,\delta)$ of Pinwheel for as small value of $\delta$ as possible.
We also want to be able to compute efficiently a feasible schedule for $V(I,\delta)$, if one exists.
By decreasing the rational frequencies $f''_i$ to integer frequencies $f_i$, we increase the density of the Pinwheel instance,
making it possibly harder to schedule. Thus we should aim at decreasing the frequencies $f''_i$ only as much as necessary.
However, simply rounding down the frequencies $f''_i$ to the nearest integers might not be the best way since the 
integer frequencies obtained that way might not be sufficiently ``regular'' to imply a feasible schedule.


\subsubsection*{\bf A 2-approximation algorithm.}
%
To give a simple illustration how Lemma~\ref{l:delta} can be used, we refer
to the result from~\cite{holte_pinwheel:_1989} which says that any instance
of Pinwheel with frequencies being powers of 2 and the density at most 1 can be scheduled and 
a feasible schedule can be easily computed.
For an instance $I$ of BGT, we take the instance $V(I,1)$ where the frequencies $f''_i$ are rounded 
down to the powers of $2$. Multiplying first the frequencies by $2$ decreases the density to $1/2$.
The subsequent rounding down to the powers of $2$ decreases each frequency less than by half, so the density of the instance
increases less than twice. Thus the final instance $V(I,1)$ of Pinwheel has the density less than $1$ and all frequencies are
powers of $2$, so it can be scheduled 
and, by Lemma~\ref{l:delta}, its feasible schedule is a $2$-approximate schedule for the original BGT instance $I$.   
In fact, this approach shows that 
$\OPT(I) \le 2H(I)$ for any BGT instance $I$.

It was shown in~\cite{fishburn_pinwheel} that every instance of Pinwheel with density
not greater than $3/4$ is feasible and, tracing the proof given in~\cite{fishburn_pinwheel},
one can obtain an efficient scheduling algorithm for such instances.
For a BTG instance $I$, 
the Pinwheel instance $V(I,\delta) = \anglebr{f_1, f_2, \ldots, f_n}$, where 
$\delta = 1/3 + h_1/H$ and 
$f_i = \lfloor (1+\delta)H/h_i \rfloor$ has density less than $3/4$ since
\[ \frac{1}{f_i} 
\: < \: \frac{1}{(1+\delta)H/h_i - 1} 
\: = \: \frac{h_i}{H}\cdot\frac{1}{1+\delta- h_i/H} 
\: \le \: \frac{h_i}{H} \cdot \frac{1}{1+\delta - h_1/H} 
\: = \: \frac{3}{4}\cdot \frac{h_i}{H}.
\]
Thus Lemma~\ref{l:delta} and Pinwheel schedules for instances with density 
at most $3/4$ give the following lemma.

\begin{lemma}\label{Approx-4-3}
There is a polynomial-time $(4/3 +h_1/H)$-approximate scheduling algorithm for the BGT problem.
\end{lemma} 

When $h_1/H$ decreases to $0$, then this approach 
gives 
approximation ratios decreasing to $4/3$,
improving on the $2$-approximation. 
We are, however, looking for an algorithm which would compute BGT schedules with 
approximation ratio decreasing to $1$ when $h_1/H$ decreases to $0$.

%

\omitthis{
\begin{proof}
Follows directly from Lemma~\ref{l:delta}.
\ep\end{proof}
}

\subsection{\bf A $(1+O(\sqrt{h_1/H}))$-approximation algorithm}\label{SubSection-approx-algo-small-rates}

\noindent
Our approach to a better BGT approximation is again through powers-of-two Pinwheel instances, 
as in the 2-approximation algorithm, but now we 
derive an appropriate powers-of-two instance through a process of gradual transformations.
We start with greater granularity of frequencies than powers of two by reducing 
the rational frequencies $f''_i = H(1+\delta)/h_i$  
to the closest values of the form $2^{k}(1+{j}/{C})$, where $k$ is an integer, $C=2^{q}$ for some
integer $q \ge 0$, and 
$j$ is an integer in $[0,C-1]$. 
We set the parameter $q$ in such a way that we always have $q \le k$, so the 
frequencies $2^{k}(1+{j}/{C})$ are integral, forming a proper instance of Pinwheel.
The values of $q$ and $C$ are not fixed constants as they depend on the input.
To show that the obtained instance $V(I,\delta)$ of Pinwheel is feasible, for a suitable choice of $\delta$,
and to construct a schedule for this instance,
we use the following two observations.

\vspace*{1ex}

\noindent
{\em Observation 1.} 
Let $V$ be an instance of Pinwheel which has two equal even frequencies $f_i = f_j = 2f$.
If the instance $V'$ obtained from $V$ by replacing these two frequencies with one frequency $f$ is feasible,
then so is instance $V$. 
Note that such updates of a sequence of frequencies do not change its density.
We obtain a schedule for $V$ from a schedule for $V'$ by replacing the occurrences of the frequency $f$
alternatingly with the frequencies $f_i$ and $f_j$.
In our algorithm we will be replacing pairs 
of equal frequencies $2^{k}(1+{j}/{C})$ with one frequency $2^{k-1}(1+{j}/{C})$.
%

\vspace*{1ex}

\noindent
{\em Observation 2.} 
Generalizing the previous observation,
let $V$ be an instance of Pinwheel which has $m$ equal frequencies $f_{i_1} = f_{i_2} =  \cdots  = f_{i_m} =  m f$, where 
$f$ is an integer.
If the instance $V'$ obtained from $V$ by replacing these $m$ frequencies with one frequency $f$ is feasible,
then so is instance $V$ and a schedule for $V$ can be easily obtained from a schedule for $V'$
(in the schedule for $V'$, for each $i \ge 0$ and $1 \le q \le m$, 
replace the $(im + q)$-th occurrence of $f$  with $f_{i_q}$). 
As before, such updates of a sequence of frequencies do not change its density. 
In our algorithm 
we will be combining $m_{j}=C+j$ frequencies $2^{k}(1+\frac{j}{C})$ into one frequency 
${2^{k}(1+{j}/{C})}/{m_{j}} = {2^{k}}/{C}$, which will 
be a power of 2. 

\vspace*{1ex}

We are now ready to describe our algorithm.
Step 1 of the algorithm initializes the Pinwheel context by changing the BGT input instance of the rates of growth 
$(h_1, h_2, \ldots, h_n)$ into 
a Pinwheel pseudo-instance $\anglebr{f''_1, f''_2, \ldots, f''_n}$ of rational frequencies.
Step 2 converts this pseudo-instance into a proper Pinwheel instance $\anglebr{f_1, f_2, \ldots, f_n}$
of integer frequencies.
Steps 3--5 transform this Pinwheel instance
into a powers-of-$2$ instance $\anglebr{g_1, g_2, \ldots, g_r}$, $r \le n$.
Steps 2--5 are illustrated in Figure~\ref{f:MainAlgo}.
The final step 6 computes first a schedule for the powers-of-2 instance $\anglebr{g_1, g_2, \ldots, g_r}$
and then expands it to a schedule
for the Pinwheel instance $\anglebr{f_1, f_2, \ldots, f_n}$, which is returned as 
the computed schedule for the BGT input instance. 

\vspace*{2.3ex}

\begin{figure}
\begin{center}
\includegraphics[scale=1.0]{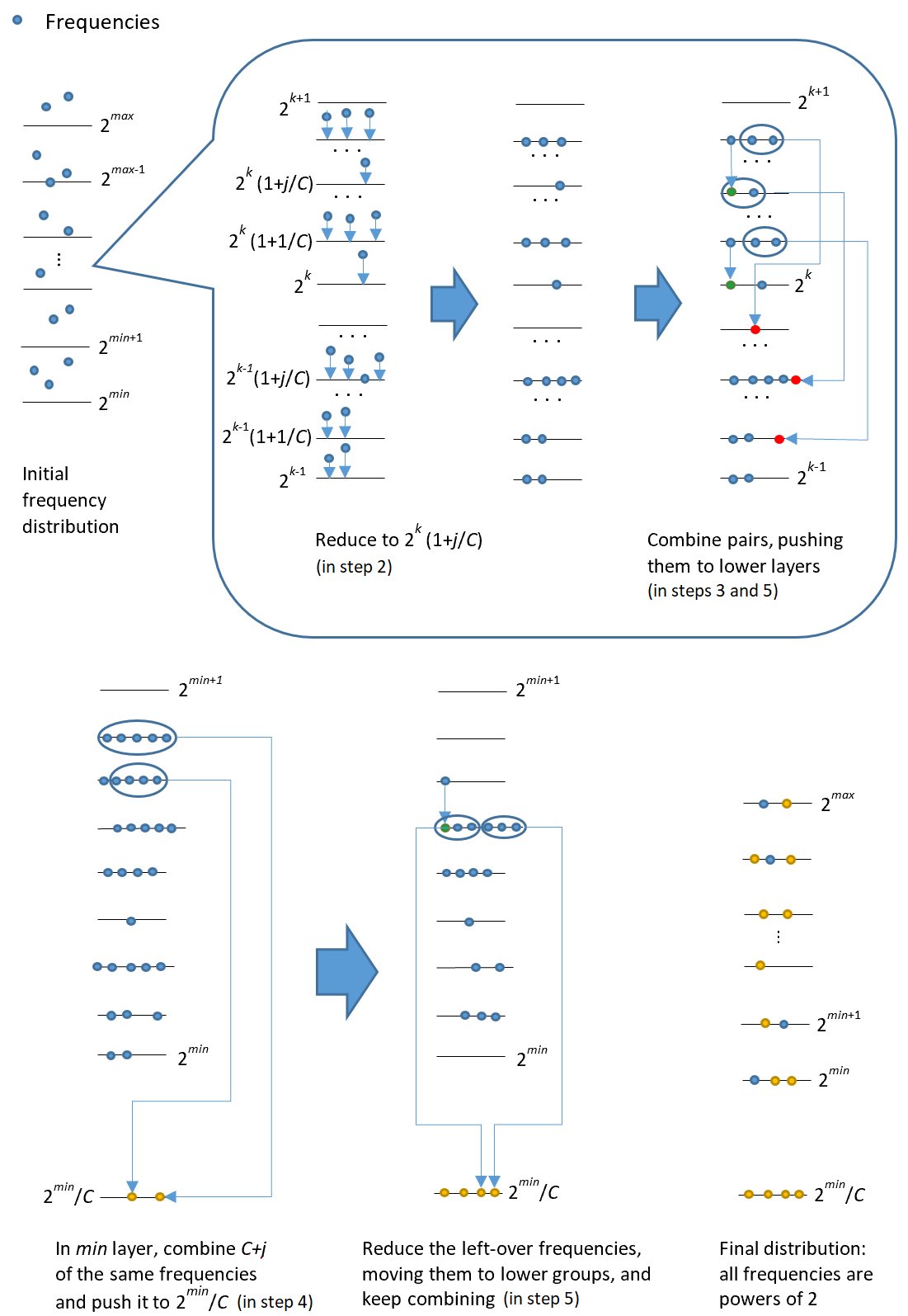}
\end{center}
\caption{Illustration of the execution of the Main Algorithm.
The top half: transformation of frequencies in layers other 
than the $min$ layer. The bottom half: transformation in the $min$ layer and the final powers-of-2 frequencies.}
\label{f:MainAlgo}
\end{figure}

\noindent
{\bf The Main Algorithm} \\[6pt]
{\bf Input}:
BGT instance $I = (h_1 \ge  h_2 \ge \dots \ge h_n)$. \\[3pt]
{\bf Output}: A (cyclic) perpetual schedule of $I$ specified by pairs $(p_i , q_i)$, for $1 \le i \le n$.
Item $b_i$ occurs in the schedule at positions $p_i + kq_i$, for $k \ge 0$.

\begin{enumerate}
\item 
Set the parameter $\delta = 3\sqrt{h_1/H} \le 3$.

 \vspace{.5ex}
 
Form a pseudo-instance $\anglebr{f''_1 \le f''_2 \le \ldots \le f''_n}$ of Pinwheel, by setting
$f''_i = (1+\delta)H/h_i$.
The density of this pseudo-instance is $\sum_{i}1/f_i'' = {1}/(1+\delta)$, 
and 
the setting of the parameter $\delta$ implies that $f''_1 \ge 4$, regardless of the value of $h_1$.
Let 
${min}\ge 2$
be the largest integer such that $2^{min} \le f''_1$
and let $max\ge min $ be the smallest integer
such that $2^{max+1} > f''_n$.
\vspace{.5ex}

\item Reduce each frequency $f''_i$ to the closest value $f_i$ of the form 
$2^{k}\left(1+\frac{j}{C}\right)$, where $k$, $C$ and $j$ are integers
such that $k\geq min$, $C = 2^q\ge 2$ for $q =\lfloor min/2\rfloor\ge 1$, and $0\le j \le C-1$.
These conditions imply that the new frequencies $f_i$ are integral.
Since the reduction of $f''_i$ to $f_i = 2^{k}\left(1+\frac{j}{C}\right)$ is by a factor less than 
$1+\frac{1}{C}$, the density of the
whole sequence of frequencies  
increases at most by a factor of $1+\frac{1}{C}$,
to at most ${\left(1+\frac{1}{C}\right)}/{(1+\delta)}$.

\vspace{.5ex}

The sequence $\anglebr{f_1, f_2, \ldots, f_n}$ is our (proper) instance $V(I,\delta)$ of the Pinwheel problem. 
The remaining steps compute a schedule of this sequence.
Steps 3-5 
use Observations~1 and~2, and further reductions of the frequencies if needed, to 
transform the sequence $\anglebr{f_1, f_2, \ldots, f_n}$ 
to a sequence $\anglebr{g_1, g_2, \ldots, g_r}$, $r \le n$ where all frequencies $g_i$ are  powers of $2$.

\vspace{.5ex}

\item 
We refer to the range $[2^k, 2^{k+1})$ of frequencies as {\em layer\/} $k$, and to the 
set of frequencies of the same value $2^{k}\left(1+\frac{j}{C}\right)$ as the {\em group} $j$ in layer $k$.

\vspace{.5ex}

For $k = max, max-1, \ldots, min+1$,
apply Observation~1 in layer $[2^{k},2^{k+1})$ 
as many times as possible to combine pairs of the same frequencies $2^{k}\left(1+\frac{j}{C}\right)$, $1 \le j \le C-1$,
pushing them down to the lower layer by replacing each such pair with one frequency $2^{k-1}\left(1+\frac{j}{C}\right)$ 
(see the top half of Figure~\ref{f:MainAlgo}).

\vspace{.5ex}

On the conclusion of this step, there is at most one frequency $2^{k}\left(1+\frac{j}{C}\right)$, for each combination of $k\in [min+1, max]$ 
and $j\in[1,C-1].$

\vspace{.5ex}

\item Apply Observation~2 in the layer $[2^{min},2^{min+1})$
until there are at most $C+j-1$ frequencies $2^{min}\left(1+\frac{j}{C}\right)$ left, for any $j\in[1,C-1]$
(see the first part of the bottom half of Figure~\ref{f:MainAlgo}).

%
%
\vspace{.5ex}

\item For $k = max, max-1, \ldots, min$, reduce all remaining frequencies in range $[2^{k},2^{k+1})$ which
are not powers of $2$, group by group starting from the top group defined by $j=C-1$, 
and pushing each frequency down to the next lower group.
While progressing down through the groups,
keep applying Observation~1 whenever possible, if in a layer $k \ge min+1$, and Observation~2 
in the lowest layer for $k = min$. 

\vspace{.5ex}

On the conclusion of this step, we have a sequence of frequencies 
$\anglebr{g_1, g_2, \ldots, g_r}$, $r \le n$,
which are powers of $2$, 
but the density further increases by some value $\Delta D\ge 0$.
Thus the density of this final powers-of-$2$ instance $\anglebr{g_i}$ of the Pinwheel problem 
is at most  ${\left(1+\frac{1}{C}\right)}/{(1+\delta)}+\Delta D$.
We will show that the setting of the parameters $\delta$ and $C$ imply that this bound is at most $1$,
so the Pinwheel sequence $\anglebr{g_i}$ is feasible.

\vspace{.5ex}

\item
Compute a cyclic schedule for the sequence of frequencies $\anglebr{g_1, g_2, \ldots, g_r}$ 
using the algorithm for powers-of-2 Pinwheel instances 
from~\cite{holte_pinwheel:_1989}.
Such a schedule is specified by pairs $(p_i , q_i)$, for $1 \le i \le r$, which mean that 
the frequency $g_i$ is placed in the perpetual schedule of $\anglebr{g_1, g_2, \ldots, g_r}$ 
at positions $p_i + kq_i$, for $k \ge 0$.

\vspace{.5ex}

Expand the schedule of $\anglebr{g_1, g_2, \ldots, g_r}$ to a schedule  of
$\anglebr{f_1, f_2, \ldots, f_n}$ 
by tracing back the applications of Observations~1 and~2. 
The schedule of $\anglebr{f_1, f_2, \ldots, f_n}$ is returned as the computed schedule of 
the BGT input instance $(h_1, h_2, \ldots, h_n)$.

%
\end{enumerate}

\begin{figure}
\begin{center}
\includegraphics[scale=0.35]{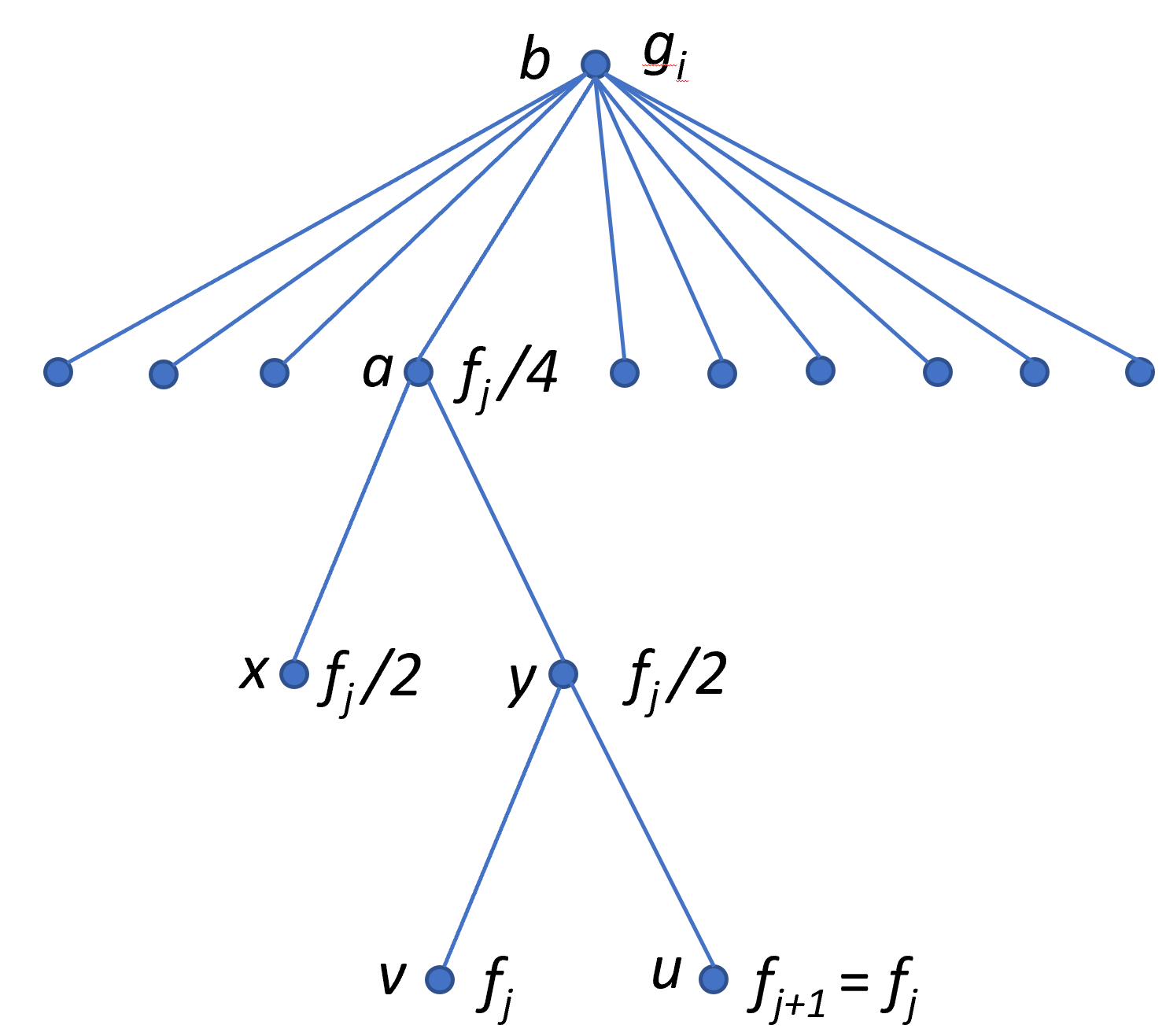}
\end{center}
\caption{Illustration of combining frequencies in the Main Algorithm.  In steps 3--5 of the algorithm,
frequency $f_j$ is paired twice with other frequencies by applications of Observation~1, creating a frequency $f_j/4$
(node $a$).
The resulting frequency is then put into a group of $10$ frequencies $f_j/4$ and this group 
is replaced with one new frequency $g_i = f_j/40$ by an application of Observation 2.
}
\label{f:CombineFreq}
\end{figure}

\vspace{.3ex}
\noindent
The final step 6 of the algorithm is illustrated by an example given in Figure~\ref{f:CombineFreq}.
In this example, a frequency $f_j$ (computed in step 2) was subsequently (in steps 3--5) paired twice with other frequencies 
by applications of Observation~1 (contributing first to a new frequency $f_j/2$ and then to a new frequency $f_j/4$)
and ended up in a group of $10$ equal frequencies $f_j/4$.
These $10$ frequencies were replaced with one new frequency $g_i = f_j/40$ by an application of Observation~2.
Let the pair $(p_i, q_i)$ represent the positions of $g_i$ in the computed cyclic schedule of  $\anglebr{g_1, g_2, \ldots, g_r}$:
frequency $g_i$ is placed at positions $p_i + kq_i$, for $k \ge 0$.

Expanding the schedule of $\anglebr{g_1, g_2, \ldots, g_r}$ to the schedule  of
$\anglebr{f_1, f_2, \ldots, f_n}$, 
every $40$-th occurrence of frequency $g_i$ 
is replaced by frequency $f_j$.
More precisely, we first replace every $10$-th occurrence of $g_i$ with the frequency $a=f_j/4$, starting 
from the $4$-th occurrence of $g_i$.
Then every second occurrence of $a$ is replaced with $y=f_j/2$, starting from the $2$-nd occurrence of $a$,
and finally every other occurrence of $y$ is replaced with $v=f_j$.
Thus frequency $f_j$ is placed in the schedule at positions $(p_i + 13 q_i) + k(40q_i)$.




\begin{theorem}\label{th:main}
For each BGT input instance, the Main Algorithm 
computes a $(1+3\sqrt{h_1/H})$-approximation schedule.
\end{theorem}

\begin{proof}
We show that 
the density of the sequence of frequencies $\anglebr{g_1, g_2, \ldots, g_r}$ computed 
in  step 5 of the algorithm is at most~$1$. Since these frequencies are powers of $2$, 
the sequence  $\anglebr{g_1, g_2, \ldots, g_r}$ is feasible
and gives a $(1+\delta)$-approximation schedule for the initial BGT input instance 
(Lemma~\ref{l:delta}).

The density of the Pinwheel pseudo-instance $\anglebr{f''_1, f''_2, \ldots, f''_n}$ is $D'' = 1/(1+\delta)$.
To bound the density $D$ of the Pinwheel instance $\anglebr{f_1, f_2, \ldots, f_n}$, observe that 
if $f_i = 2^{k}\left(1+{j}/{C}\right)$, then $f_i$ has been obtained from the rational frequency $f_i''$
such that 
\[ 2^{k}\left(1+\frac{j}{C}\right) \le f''_i < 2^{k}\left(1+\frac{j+1}{C}\right).
\]
This means that 
\[ f_i = 2^{k}\left(1+\frac{j}{C}\right) 
  > f''_i \left(1+\frac{j}{C}\right)/\left(1+\frac{j+1}{C}\right) = f''_i \frac{C+j}{C+j+1} \ge f''_i \frac{C}{C+1},
\]
so 
\[D < (1+1/C) D'' = (1+1/C)/(1+\delta).
\]

Steps 3 and 4 of the algorithm transform the Pinwheel instance $\anglebr{f_1, f_2, \ldots, f_n}$
using Observations 1 and 2, but without changing its density.

Step 5 further transforms the sequence of frequencies, increasing the density 
of the sequence whenever individual frequencies are reduced.
We bound separately the increase $\Delta D_{above}$
of the density when we modify the frequencies in layers 
$max, max -1, \ldots , min +1$,
and the increase $\Delta D_{min}$ of the density when we modify the frequencies in layer $min$.

\begin{eqnarray*}
\Delta D_{above} 
 & \le & \sum_{k=min+1}^{max} 
     \sum_{j=1}^{C-1} \frac{1}{2^k}\left( \frac{1}{1+(j-1)/C} - \frac{1}{1+j/C} \right) \\
& = & \sum_{k=min+1}^{max} 
     \frac{1}{2^k}\left( 1 - \frac{1}{1+(C-1)/C} \right)
     \le \frac{1}{2^{min}} \cdot \frac{C-1}{2C-1}\,.
\end{eqnarray*}

\begin{eqnarray*}
\Delta D_{min} 
 & \le & \sum_{j=1}^{C-1} (C+j-1) \frac{1}{2^{min}} 
              \left( \frac{1}{1+(j-1)/C} - \frac{1}{1+j/C} \right) \\
& = & \frac{C}{2^{min}} \cdot \sum_{j=1}^{C-1}  \frac{1}{C+j}\,.
\end{eqnarray*}

Let $K = 2^{min} / C^2 \in\{1,2\}$.
Then the density of the final powers-of-two Pinwheel instance $\anglebr{g_1, g_2, \ldots, g_r}$ is at most

\begin{eqnarray} 
\lefteqn{\hspace{-4mm}D + \Delta D_{above} + \Delta D_{min}} \nonumber \\[1mm] 
  & \le & \frac{1+1/C}{1+\delta} + \frac{1}{2^{min}} \cdot \frac{C-1}{2C-1}
     + \frac{C}{2^{min}}  \cdot \sum_{j=1}^{C-1}  \frac{1}{C+j} \nonumber \\
   & = & \frac{1+1/C}{1+\delta} + \frac{1}{KC} \cdot \left(\frac{C-1}{C\cdot (2C-1)}
     + \sum_{j=1}^{C-1}  \frac{1}{C+j}\right) \label{kldf89a} \\[1mm]
   & = & \frac{1+1/C}{1+\delta} + \frac{1}{KC} \cdot \left(\frac{1}{C} - \frac{1}{2C-1}
     + \sum_{j=1}^{C-1}  \frac{1}{C+j}\right) \nonumber \\[1mm]
   & = & \frac{1+1/C}{1+\delta} + \frac{1}{KC} \cdot \sum_{j=C}^{2C-2}  \frac{1}{j} \nonumber \\
   & \le & \frac{1+1/C}{1+\delta} + \frac{1}{KC} \cdot \ln 2 \label{kldf89b} \\[1mm]
   & = & \frac{1}{1+\delta} + \frac{1}{C} \cdot \left(\frac{1}{1+\delta}+\frac{\ln 2}{K}\right) \nonumber \\[1mm]
   & \le & \frac{1}{1+\delta} + \frac{\sqrt{2K}}{3} \cdot \frac{\delta}{\sqrt{1+\delta}} \cdot \left(\frac{1}{1+\delta}+\frac{\ln 2}{K}\right) \label{chjk1a} \\[1mm]
   & \le & 1. \label{jkk34s}
\end{eqnarray}
To get Equality~\eqref{kldf89a}, use $2^{min} = K C^2$.
To get Inequality~\eqref{kldf89b}, use the known fact that for the harmonic numbers $H_k = 1 + 1/2 + 1/3 + 1/4 + ... + 1/k$, we have $H_k = \ln k  + \delta_k$, where $(\delta_k)$ is a sequence of positive numbers strictly monotonically decreasing. Hence, we have $H_{2k} - H_k = \ln 2 + \delta_{2k} - \delta_k < \ln 2$.
To get Inequality~\eqref{chjk1a}, use 
\[
C \cdot \sqrt{2K}  = \sqrt{2^{min+1}}                  
> \sqrt{f''_{1}} =  \sqrt{(1+\delta)\frac{H}{h_1}} 
= \frac{3\sqrt{1+\delta}}{\delta}\,. 
\]
Inequality~\eqref{jkk34s} follows from the fact that for $0 \le \delta \le 3$ and $K \in \{1,2\}$,
the function $f_K(\delta)$ in line~\eqref{chjk1a} is maximized for $\delta=0$.
To see this, substitute $\gamma \equiv\sqrt{\delta + 1}$ (so for $\delta$ increasing from $0$ to $3$, 
$\gamma$ increases from $1$ to $2$) to obtain
\begin{eqnarray}
 f_K(\delta) - 1 =  f_K(\gamma^2 -1) - 1
 & = & \frac{1}{\gamma^2} + \frac{\sqrt{2K}}{3} \cdot \frac{(\gamma^2 -1)}{\gamma} 
           \cdot \left(\frac{1}{\gamma^2}+\frac{\ln 2}{K}\right)-1  \nonumber\\
& = & \frac{\gamma^2 -1}{3\gamma^3} 
         \left( \frac{\sqrt{2}\ln 2}{\sqrt{K}} \gamma^2 -3\gamma +  {\sqrt{2K}} \right).
\label{opsds3a}
\end{eqnarray}
The quadratic in the parentheses in~\eqref{opsds3a}
has two distinct real roots
\[ \frac{3\pm \sqrt{9-8\ln 2}}{2\sqrt{2}\ln 2/\sqrt{K}}.
\]
For both $K\in \{1,2\}$, the smaller root $r_{K,1}$ is less than $1$
($r_{1,1} < r_{2,1} = 0.823...$)
and the larger root $r_{K,2}$ is greater than $2$
($r_{2,2} > r_{1,2} = 2.478...$).
This means that~\eqref{opsds3a} is equal to $0$ for $\gamma=1$ and negative for $\gamma\in (1,2]$,
or equivalently,  $f_K(\delta) - 1$ is equal to $0$ for $\delta = 0$ and negative for $\delta\in(0,3]$.
\ep\end{proof}

The discussion in Section~\ref{sec:BGT-Pinwheel} and Theorem~\ref{th:main} imply the following corollary.

\begin{corollary}
For Pinwheel instances $2 \le f_{1} \le f_{2}\le \ldots \le f_{n}$, 
the Main Algorithm applied to sequences $(1/f_{1}, 1/f_{2}, \ldots,  1/f_{n})$ is a 
Pinwheel scheduler with the guarantee density threshold $1 - 3/\sqrt{f_1}$. 
Thus each Pinwheel instance with density at most $1 - 3/\sqrt{f_1}$ is feasible. 
\end{corollary}

\begin{theorem}\label{th:main2}
The Main Algorithm can be implemented so that its running time is $O(n\log n)$, 
assuming constant-time operations on real numbers, including the logarithm operation and 
rounding. 
\end{theorem}

\begin{proof}
Steps 1 and 2 take linear time (or $O(n\log n)$, if the input sequence is not provided in the sorted order).
In steps 3--5,
the non-empty groups  are maintained in a dictionary data structure, which can be created in $O(n\log n)$ time.

In steps 3 and 4, the groups are considered in the decreasing order (from the highest group in the highest layer).
In step 3, if two frequencies are paired by Observation 1, then the group to which the new frequency should be added 
can be found (or created, if empty) in $O(\log n)$ time. 
This can happen only $O(n)$ times, since each pairing reduces the number of frequencies, so step 3 
takes $O(n\log n)$ time.
In step 4, if $m$ frequencies are combined by Observation 2 into one new frequency $2^{min}/C$, 
then the required computation takes $O(m)$ time (in this case, 
no need to look for the group with frequencies $2^{min}/C$).
Thus step 4 takes $O(n)$ time.

In step 5, the groups are again considered in the decreasing order. 
For each group, frequencies are paired by Observation 1 or combined by Observation 2 (as above, all this computation 
takes $O(n\log n)$ time) and the remaining frequencies are appended to the next non-empty group in constant time.
Thus step 5 takes $O(n\log n)$ time.

Step 6, the computation of a cyclic schedule of the feasible powers-of-2 frequencies $\anglebr{g_1, g_2, \ldots, g_r}$,
is done by the following algorithm from~\cite{holte_pinwheel:_1989}, which can be implemented to run in linear time,
assuming that the input sequence is sorted. 
Take two largest frequency values $g_{r-1}$ and $g_{r}$,
replace them with a new frequency $g' = 2g_{r-1}$ and compute a cyclic schedule for the new $(r-1)$-element 
sequence. 
To get a schedule for the original sequence $\anglebr{g_1, g_2, \ldots, g_r}$, replace the occurrences of $g'$ with 
alternating occurrences of $g_{r-1}$ and $g_{r}$.
That is, if  $g'$ occurs at positions $p + k q$, for $k\ge 0$, then $g_{r-1}$ occurs at positions $p + 2k q$ and
 $g_{r}$ occurs at positions $p + (2k+1)q$, for $k\ge 0$.
 
During the computation in steps 3--5, the pairing and combining of frequencies is recorded in ordered trees 
as illustrated in Figure~\ref{f:CombineFreq}. The additional time required to create one such tree is linear in the 
size of the tree, so $O(n)$ time for creating all trees.
By traversing these trees, we can  
expand in $O(n)$ time the cyclic schedule of $\anglebr{g_1, g_2, \ldots, g_r}$ to a cyclic schedule of $\anglebr{f_1, f_2, \ldots, f_n}$.

The Main Algorithm computes the representation of a cyclic schedule in the form of pairs $(p_i, q_i)$, which enables
constant-time answers to queries: given the current step $T$, when will be the next step when a given element $b_i$ 
is serviced?
To generate the schedule step by step, taking $O(\log n)$ time per step, 
maintain a priority queue (e.g. as the heap data structure) 
with the pairs $(i, t_i)$, where $i = 1,2,\ldots, n$ and $t_i$ is the next step when element $b_i$ will be serviced.
 Initially the priority queue contains pairs $(i, p_i)$. 
 For the current step of the schedule, remove from the priority queue the pair $(i,t_i)$ with the smallest $t_i$, 
 schedule the element $b_i$, and insert to the priority queue the pair $(i, t_i+q_i)$.\footnote{%
This works, if the cyclic schedule does not have gaps, and
the cyclic schedule computed in the Main Algorithm does not have gaps.
If the cyclic schedule may have gaps, then we should keep track of the current step $T$ and schedule $b_i$ only 
if $t_i = T$.}\qed
\end{proof}

\section{Improved approximation for the general case}
\label{sec:impoved-approx}

In this section, we consider general approximation (upper) bounds of polynomial 
BGT algorithms, that is, bounds on the approximation ratios which hold for any BGT input, irrespectively of its characteristics.
Comparing with the best previous general bound
of $12/7$ given in~\cite{Ee21}, 
we show a BGT algorithm with approximation ratio
not greater than $8/5+o(1) < 12/7$.
We assume throughout this section that the growth rates
are normalized to $H=1$. 

Our algorithm uses the approximation algorithm 
from Section~\ref{SubSection-approx-algo-small-rates}, which computes $(1+o(1))$-approximation 
schedules, if the highest growth rate $h_1$ 
is~$o(1)$.
The algorithm first partitions the input set $I$ into two groups $S$ and $L$, separating bamboos with small growth rates of values $o(1)$ 
from bamboos with large growth rates. 
Then schedules ${\cal S}$ for $S$ and ${\cal L}$ 
for $L$ are computed separately, and finally these 
two schedules are merged into the final schedule 
${\cal I}$ for the whole input $I$.
The schedule ${\cal S}$ for the small rates is computed 
using the approximation algorithm 
from Section~\ref{SubSection-approx-algo-small-rates}, so $\MH({\cal S}) \le H(S) + o(1)$.
The schedule ${\cal L}$ for the large rates is either an optimal schedule, computed
using exhaustive search, or an approximation schedule computed by 
the 2-approximation algorithm which we mentioned in Section~\ref{s:pinwheel}.
The choice 
depends on the parameters of $L$ and $S$.
In the former case, we use the bound 
$\MH({\cal L}) = \OPT(L) \le \OPT(I)$
in the analysis of the approximation ratio, while in the latter case, 
we use the bound $\MH({\cal L}) \le 2\cdot H(L)$.

We summarise this result in the theorem stated below
and the remaining part of this section is the proof of this theorem.
In our case analysis, for some cases we need to separate the fastest growing bamboo $b_1$ from the set $L$, getting a partition of $I$ into three sets
$B=\{b_1\}$, $L$ and $S$, and
computing schedules ${\cal S}$ and ${\cal L}$ for $S$ and $L$ as above. 
The final schedule ${\cal I}$ for the whole input $I$ is defined by a pattern in which the schedules 
${\cal S}$ and ${\cal L}$, and cutting $b_1$ (if $b_1$ is separated from $L$) are interleaved.
For example, if ${\cal I}$ is defined by the pattern $(L,B,L,S)$, then the schedule ${\cal L}$ is put in ${\cal I}$ in every other
position starting from position~$1$, the schedule ${\cal S}$ is in every fourth position starting from position~$4$, and $b_1$ occupies the remaining positions. 
%




\begin{theorem}
There is a polynomial-time algorithm which 
for any instance $I$ of the BGT problem,
computes a $(8/5+o(1))$-approximation schedule.
\end{theorem}


\noindent
\paragraph{The basic split of the input sequence.} 
We partition $I$ into two groups 
of large growth rates
$L=\{b_i:h_i\ge 1/m\}$
and small growth rates 
$S=\{b_i:h_i< 1/m\}$.
We set the threshold value $m = \log n/ (4\log\log n)$, where 
$\log x$ stands here for $\log_2 (x)$, if $x > 1$, or for $1$, if $x \le 1$ 
(so $m$ is well-defined and positive for all integers $n \ge 2$).
Observe that $m < n$, so $S$ is never empty, but $L$ can be empty, if there are no large growth rates.
We define 
$\bar{L} = H(L) = \sum_{b_i\in L} h_i$
and $\bar{S}= H(S)$, so
$\bar{L}+\bar{S}=1$.
Generally, if $A$ is a subset of bamboos,
then 
$\bar{A}$ denotes the sum of the rates of bamboos in~$A$.
Let ${\cal S}$ be a schedule for $S$ computed by
the approximation algorithm 
from 
Section~\ref{SubSection-approx-algo-small-rates}.
If $\bar{S}=\Omega(1)$, then $\MH({\cal S}) = \bar{S}(1+o(1))$, 
and if $\bar{S}=o(1),$ then $\MH({\cal S}) \le 2\bar{S} = o(1)$,
so in both cases $\MH({\cal S}) = \bar{S}+ o(1)$.

\begin{fact}\label{opt}
An optimal schedule
for the growth rates $L$ can be computed 
in $o(n)$ time. 
\end{fact}
\begin{proof}
The lower bound of $1/m$ on the growth rates in $L$
imply that $|L|\le m$.
Since $\OPT(L) \le 2\bar{L} \le 2$, then 
in an optimal schedule for $L$, every bamboo can have at most $2m$ different heights. 
This means that the number of configurations to consider is at most $(2m)^m$, where a configuration 
is an $|L|$-tuple of possible heights of the bamboos in $L$, as observed just before the next cutting.
Each (infinite) schedule $\ell = (i_1, i_2, \ldots)$ can be represented as a sequence of configurations
$(C_1, C_2, \ldots)$.
If this sequence is $(C_1, C_2, \ldots, C_k, C_{k+1}, \ldots, C_{k+p}, \ldots)$,
where $C_{k+1} = C_{k+p}$ is the first repeat of the same configuration, 
then for the periodic schedule $\ell^*$ formed as $(i_1, i_2, \ldots, i_k)$ followed by repetitions of
$(i_{k+1}, \ldots, i_{k+p})$, we have $\MH(\ell)^* \le \MH(\ell)$.
Hence, there is an optimal schedule which is periodic, where $k$ (the length of the initial sequence of 
distinct configurations)
and $p$ (the length of the period) are both at most $(2m)^m$.
Such a periodic optimal schedule can be found by an exhaustive search, or any better search strategy which 
guarantees taking into account all periodic schedules.

This type of argument -- a configuration space of finite size and finitely many feasible structures --
has been used previously in analyses of processes of similar nature.
For example, in the context of computing a perpetual schedule of maintaining machines, 
such an argument was used by 
Anily~{\it et al.}~\cite{DBLP:journals/dam/AnilyGH98}.

We provide further details to justify our choice of $m$.
We have set the value of $m$ 
appropriately small in terms of $n$, so that the following algorithm finds an optimal schedule for 
$L$ in $o(n)$ time. 
Let $G$ be the directed graph of the configurations space: the vertices are the configurations 
with height at most $2|L|$ (that is, with the maximum height of bamboo at most $2|L|$),
and each edge is
a possible one-round change from the current configuration to another one, defined by cutting a given bamboo.
This graph has at most $(2m)^m$ vertices and each vertex has at most $|L| \le m$ outgoing edges 
(at each configuration, there are $|L|$ choices of a bamboo for cutting, but some choices may take the maximum height of a bamboo above the $2|L|$ threshold).
For any positive number $M$, $\OPT(L) \le M$, if and only if, 
the subgraph of $G$ induced by all configurations with height at most $M$ has a cycle reachable from the initial 
configuration $(h_j: j \in L)$ (the configuration just before the first cutting).
Thus we can find $\OPT(L)$, and an optimal schedule for $L$, by binary search over
the $O(m^2)$ possible heights of
configurations (each bamboo has at most $2m$ possible heights). 
In this binary search, one iteration takes time linear in the size of graph $G$,
so the total running time is $O((2m)^m m \log m) = O(n^{1/2})$.
%
%
%
\qed
\end{proof}

\begin{fact}\label{apr}
For the set $S$, if $\bar{S}=\Omega(1),$ 
then the algorithm of
Section~\ref{SubSection-approx-algo-small-rates}
computes a $(1+o(1))$-approximation schedule
$\cal S$ for $S$
in $O(n\log n)$ time.
\end{fact}
\begin{proof}
A direct consequence of 
Theorems~\ref{th:main}
and~\ref{th:main2}.\qed
\end{proof}

If $L = \emptyset$, then we have a $(1+o(1))$-approximation for the whole instance~$I$ from Fact~\ref{apr}, so from now on we assume
that $L\not= \emptyset$.
We consider separately the cases specified below. The $(8/5 +o(1))$-approximation factor follows because for each case and 
each $b\in I$, we either establish that the maximum height of $b$ in ${\cal I}$ is at most $8/5 +o(1)$, 
or we establish that it is at most $\OPT(I)(8/5 +o(1))$.

\begin{description}
\item[Case 1:] $|L| \ge 2,\;$ $0 < \bar{S}\le 2/5,$
\item[Case 2:] $|L| \ge 2,\;$ $2/5 < \bar{S}\le 8/15,\;$ $h_1 \le 8/25$,
\item[Case 3:] $|L| \ge 2,\;$ $2/5 < \bar{S}\le 8/15,\;$ $h_1 > 8/25$,
\item[Case 4:] $|L| \ge 2,\;$ $8/15< \bar{S}\le 3/5,$
\item[Case 5:] $|L| \ge 1,\;$ $3/5 < \bar{S} < 1$,
\item[Case 6:] $|L| = 1,\;$ $0 < \bar{S} \le 3/5$.
\end{description}

\noindent
{\bf Case 1:} $|L| \ge 2,$ $0 < \bar{S}\le 2/5$. 
We use an optimal schedule ${\cal L}$ for $L$ and the pattern $(L,L,L,S)$.
We obtain a schedule ${\cal I}$ where the bamboos in $S$ stay within the height 
$4 \cdot (\bar{S} + o(1)) \le 8/5 + o(1)$.  
For any $b\in L$, if $f$ denotes the frequency of $b$ in ${\cal L}$, that is, the longest time, in rounds, between two consecutive
cuts of $b$ in ${\cal L}$, then the frequency of $b$ in ${\cal I}$ is at most 
$f+\lceil f/3 \rceil$ 
(for two consecutive cuts of $b$ 
in $\cal L$, we are adding between them in ${\cal I}$ at most $\lceil f/3 \rceil$ cuts from sequence ${\cal S}$).
Thus the height of $b$ in ${\cal I}$ is never greater than
\begin{eqnarray}
(f+\lceil f/3 \rceil)h_b &=& fh_b \left( 1 + \frac{\lceil f/3 \rceil}{f} \right)
\le \OPT(L)\left( 1 + \frac{\lceil f/3 \rceil}{f} \right)  \nonumber \\
&\le& \OPT(I)\left( 1 + \frac{\lceil f/3 \rceil}{f}  \right)
\le \frac{3}{2}\OPT(I). \label{kjlklew3w}
\end{eqnarray}
The last inequality holds because for $f \ge 2$, $\lceil f/3\rceil \le f/2$.

\smallskip
\noindent
{\bf Case 2:} $|L| \ge 2,$ $2/5 < \bar{S}\le 8/15,$ $h_1 \le 8/25$.
We use an optimal schedule ${\cal L}$ for $L$ and the pattern $(L,L,S)$.
For any $b\in S$, its height in ${\cal I}$ is at most 
$3 \cdot (\bar{S} + o(1)) \le 8/5 + o(1)$.
For any $b\in L$, its frequency $f$ in ${\cal L}$ increases to at most 
$f+\lceil f/2 \rceil$ in ${\cal I}$, since now for two consecutive cuts of $b$ 
in $\cal L$, we are adding in ${\cal I}$ at most $\lceil f/2 \rceil$ cuts from sequence ${\cal S}$. 
Thus, following a similar argument as in~\eqref{kjlklew3w}, the height of $b$ in ${\cal I}$ is at most
$\OPT(I)(1+ \lceil f/2 \rceil/f)$. The factor $1+ \lceil f/2 \rceil/f$ is at most $8/5$ for each $f\ge 2$ except $f=3$
(we have $\lceil f/2 \rceil \le (3/5)f$ for each integer $f\ge 2$ except $f=3$).
If $f=3$, then the height of $b$ in ${\cal I}$ is at most
$(f+ \lceil f/2 \rceil) h_1 = 5 h_1 \le 8/5$.

\smallskip
\noindent
{\bf Case 3:} $|L| \ge 2,$ $2/5 < \bar{S}\le 8/15,$ $h_1 > 8/25$.
We remove $b_1$ from $L$ and put it in a separate singleton set $B$, 
and we move some bamboos from $S$ to $L$, to obtain a partition of $I$ into three sets 
$B =\{b_1\}$, $S'$ and $L'$ such that $\bar{S'} = 2/5 + o(1)$
and $\bar{L'} \le 3/5 - h_1$. This is feasible, because all growth rates in $S$ are $o(1)$.
The set $L'$ may now be large, with size up to linear in $n$, but from this case on, 
we do not compute an optimal schedule for $L'$,
using instead a schedule ${\cal L'}$
computed by the 2-approximation algorithm mentioned in Section~\ref{s:pinwheel}, which runs in $O(n\log n)$ time.
For set $S'$, we use a schedule ${\cal S'}$ computed by
the approximation algorithm 
from 
Section~\ref{SubSection-approx-algo-small-rates}.

The final schedule ${\cal I}$ is defined by the pattern $(L',B,L',S')$, 
if $8/25 < h_1 \le 2/5$, or the pattern $(B,L',B,S')$, 
if $h_1 > 2/5$.
In both cases, the height of any $b\in S'$ in ${\cal I}$ is at most $4(2/5 + o(1))= 8/5 + o(1)$.
In the former case, for any $b\in L'$, the height of $b$ in ${\cal L'}$ is at most $2\bar{L'}$, so 
at most $4\bar{L'} \le 4(3/5 - h_1) \le 28/25 < 8/5$ in ${\cal I}$, and 
the height of $b_1$ is at most $4 h_1 \le 8/5$.
In the latter case, for any $b\in L'$, the height of $b$ in ${\cal L'}$ is at most $2\bar{L'}$, so 
at most $8\bar{L'} \le 8(3/5 - h_1) \le 8/5$ in ${\cal I}$, and
the height of $b_1$ is at most $2 h_1 \le \OPT(I)$.

\smallskip
\noindent
{\bf Case 4:} $|L| \ge 2,$ $8/15 < \bar{S}\le 3/5$.
We move some bamboos from $S$ to $L$, to obtain a partition of $I$ into two sets 
$S'$ and $L'$ such that $\bar{S'} = 8/15 + o(1)$,
and $\bar{L'} \le 7/15$, and we compute schedules ${\cal S'}$ and ${\cal L'}$ as in the previous case.
The final schedule ${\cal I}$ is defined by the pattern $(L',L',S')$. 
The height of any $b\in S'$ in ${\cal I}$ is at most $3(8/15 + o(1))= 8/5 + o(1)$.
As in Case 2, the frequency of any $b\in L'$ increases from $f$ in ${\cal L}$ to at most 
$f+\lceil f/2 \rceil$ in ${\cal I}$, so the maximum height of $b$ increases from $fh_b \le 2\bar{L'} \le 14/15$ 
in ${\cal L}$ to at most 
$(f+\lceil f/2 \rceil)h_b \le (14/15)(1+\lceil f/2 \rceil/f) \le 8/5$ in ${\cal I}$.
For the last inequality, 
recall from Case 2 that $\lceil f/2 \rceil \le (3/5)f$ for 
all $f \ge 2$ except $f=3$, and verify the $f=3$ case separately.

\smallskip
\noindent
{\bf Case 5:} $|L| \ge 1,$ $3/5 < \bar{S} < 1$.
Move some bamboos from $S$ to $L$, to obtain a partition of $I$ into two sets 
$S'$ and $L'$ such that $\bar{S'} = 3/5 + o(1)$
and $\bar{L'} \le 2/5$, compute schedules ${\cal S'}$ and ${\cal L'}$ as in the previous two cases,
and define the final schedule ${\cal I}$ by the pattern $(L',S')$. 
The height of any $b\in S'$ in ${\cal I}$ is at most $2(3/5 + o(1)) < 8/5$, 
and the height of any $b\in L'$ in ${\cal I}$ is at most $4\bar{L'} \le 8/5$. 

\smallskip
\noindent
{\bf Case 6:} $|L| = 1,$ $0 < \bar{S} < 3/5$.
The set $I$ is partitioned into $B=\{b_1\}$ and $S = I - B$, and the final schedule ${\cal I}$ is defined by $(B,S)$.
The height of $b\in S$ in ${\cal I}$ is at most $2(3/5 + o(1)) < 8/5$, 
and the height of $b_1$ is at most $2h_1 \le OPT(I)$. 

\smallskip
The running time of this approximation algorithm is dominated 
by the $O(n\log n)$ time of the approximation algorithm from Section~\ref{SubSection-approx-algo-small-rates}.


%
%
\section{Continuous BGT}\label{sec:ContinuousBGT}

We consider now the continuous variant of the BGT problem.
Since this variant models scenarios when bamboos are spread over some geographical area,
we will now refer not only to bamboos $b_1, b_2, \ldots , b_n$
but also to the points $v_1, v_2, \ldots , v_n$
(in the implicit underlying space) where these bamboos are located.
We will denote by $V$ the set of these points.

Recall that input $I$
for the continuous BGT problem consists of the growth rates $(h_i:\: 1\le i \le n)$
and the travel times between bamboos $(t_{i,j}:\: 1\le i,j\le n)$.
We assume that $h_1 \geq h_2 \geq \ldots \geq h_n$, as before,
and normalize these rates, for convenience, so that $h_1+h_2+ \ldots +h_n = H= 1$
(this normalization is done without loss of generality, since only the relative heights of bamboos are relevant).
We assume that the travel distances form a metric on~$V$.

For any $V' \subseteq V$,
the minimum growth rate among all points in $V'$ is denoted by $h_{\min}(V')$,
and the maximum growth rate among all points in $V'$ is denoted
by $h_{\max}(V')$. Let $h_{\min} = h_{\min}(V) = h_n$, and $h_{\max} = h_{\max}(V) = h_1$.
Recall that we assume $n\ge 2$ (to avoid the trivial case), and 
now we additionally assume that $h_{\max} > h_{\min}$ (the uniform case, when $h_{\max} = h_{\min}$, is TSP).

The diameter of the set $V$ is denoted by $D = D(V) = \max\{t_{i,j}:\: 1 \le i, j \le n\}$.
For any $V' \subseteq V$, $\MST(V')$ denotes the minimum weight of a spanning tree on $V'$
(the travel times are the weights of the edges).
Recall that for an algorithm $\cA$ and input $I$,
$\MH(\cA(I))$ denotes the maximum height that any bamboo ever reaches, if
trimming is done according to the schedule computed by $\cA$,
and $\OPT(I)$  is the optimal (minimal) maximum height of a bamboo over all schedules.

\subsection{Lower bounds and simple approximation based on discrete BGT}\label{sec:Continuous-LB}

We first show some simple lower bounds on the maximum height of a bamboo.
For notational brevity, we omit the explicit reference to the input $I$.
For example, the inequality $\MH(\cA) \ge D h_{\max}$ in the lemma below
is to be understood as $\MH(\cA(I)) \ge D(V(I)) \cdot h_{\max}(V(I))$, for each input instance $I$.

\begin{lemma}\label{LB1}
$\MH(\cA) \ge D h_{\max}$, for any algorithm $\cA$.
\end{lemma}

\begin{proof}
The robot must visit another point $x$ in $V$ at distance at least $D/2$ from $v_1$.
When the robot comes back to $v_1$ after visiting $x$ (possibly via a number of other points in $V$),
the bamboo at $v_1$ has grown at least to the height of $D h_1 = D h_{\max}$.
\ep\end{proof}

\begin{lemma}\label{LB2}
$\MH(\cA) \ge h_{\min}(V') \cdot \MST(V'),$ for any algorithm $\cA$ and $V' \subseteq V$.
\end{lemma}

\begin{proof}
Let $v$ be the point in $V'$ visited last: all points in $V'\setminus\{v\}$ have been visited
at least once before the first visit to $v$.
The distance traveled until the first visit to $v$ is at least $\MST(V')$,
so the bamboo at $v$ has grown to the height at least $h_v \cdot \MST(V')\ge h_{\min}(V') \cdot \MST(V')$.
\ep\end{proof}

One may ask how good are the schedules for continuous BGT which are computed taking into account only 
the growth rates, ignoring the travel times. If we use, for example, the schedules computed by 
the 2-approximation algorithm from Section~\ref{sec:DiscreteBGT-offline},
which guarantee that the maximum height of a bamboo does not grow above $2$
(recall that we normalize the growth rates, so $H = 1$), 
then there are at most $\lfloor 2/ h_i \rfloor - 1$
cutting actions
between two consecutive cuttings of bamboo $b_i$.
Otherwise bamboo $b_i$ would grow in discrete BGT to the height at least 
$h_i (\lfloor 2/ h_i \rfloor +1) > 2 = 2H$.
Thus in continuous BGT the time between two consecutive cuttings of 
bamboo $b_i$ is at most $D\lfloor 2/ h_i \rfloor$, so the height of this bamboo is never greater than 
$h_i D \lfloor 2/ h_i \rfloor \le 2D$.
Combining this with Lemma~\ref{LB1},
we conclude that the 2-approximation algorithm for discrete BGT 
is a $(2/h_{\max})$-approximation algorithm for continuous BGT.
In particular, this approach gives $\Theta(1)$-approximate algorithm for continuous BGT
for inputs with $h_1 = \Theta(1)$. 
To derive good approximation algorithms for other types of input,
we will use the lower bound from Lemma~\ref{LB2}.

\subsection{Approximation algorithms}\label{sec:ContinuousAlgo}


We present in this section three approximation algorithms for continuous BGT which
are based on computing spanning trees.
Algorithm 1 computes only one spanning tree of all points and the schedule for the robot 
is to traverse repeatedly an Euler tour of this tree. 
This simple algorithm ignores the growth rates of bamboos and computes a schedule of cutting
taking into account only the travel times between points.

Algorithms 2 and 3 group the bamboos according to the similarity of their growth rates 
and compute a separate spanning tree 
for each group. 
The robot moves from one tree to the next in the Round-Robin fashion.
At each spanning tree $T$, the robot walks along the Euler-tour of this tree for time~$D$ before moving 
to the next tree. When the robot eventually comes back to tree $T$, it resumes traversing 
the Euler tour of $T$
from the point when it last left this tree.
The growth rates of bamboos in the same spanning tree are within a factor of~2.
Algorithm 3 differs from Algorithm 2 by using spanning trees only for the first $\Theta(\log n)$ groups, which 
include bamboos with growth rates greater than $1/n^2$. 
The remaining bamboos, which all have low growth rates, are dealt with individually.

We describe our Algorithms 1, 2 and 3 in pseudocode
and give their approximation ratio in the theorems below.
The description of each algorithm consists of two parts: \emph{pre-processing} and \emph{walking}. 
We do not explicitly mention the actions of cutting/attending bamboos, assuming that 
whenever the robot passes a point in $V$\!, it cuts the bamboo growing at this point.

\begin{figure}[t]
\vspace{-.2cm}
\hrule\vspace{.1cm}
\vspace{3mm}
\textbf{Algorithm 1:} An $O({h_{\max}}/{h_{\min}})$-approximation algorithm
for continuous BGT.
\begin{enumerate}
\item
{[\emph{Pre-processing}]}
Calculate a minimum spanning tree $T$ of the point set $V$.
\item
{[\emph{Walking}]}
Repeatedly perform an Euler-tour traversal of $T$.
\end{enumerate}
\hrule
\end{figure}

\begin{theorem}
Algorithm~1 is an $O({h_{\max}}/{h_{\min}})$-approximation algorithm for the continuous BGT problem.
\end{theorem}

\begin{proof}
Let $\cA_1$ denote Algorithm 1.
Every point $v_i \in V$ is visited by $\cA_1$ at least every $2 \cdot \MST(V)$ time units.
Hence,
\begin{equation}
\label{eq:algo1a}
\MH(\cA_1) = O(h_{\max}(V) \cdot \MST(V)).
\end{equation}
According to Lemma~\ref{LB2},
\begin{equation}
\label{eq:algo1b}
\OPT = \Omega(h_{\min}(V) \cdot \MST(V)).
\end{equation}
Combining the two bounds (\ref{eq:algo1a}) and (\ref{eq:algo1b}),
it follows that Algorithm~1 is an $O({h_{\max}}/{h_{\min}})$-approximation algorithm for BGT.
\ep\end{proof}

\begin{figure}[t]
\vspace{-.2cm}
\hrule\vspace{.1cm}
\vspace{3mm}
\textbf{Algorithm 2:} An $O(\log\lceil {h_{\max}}/{h_{\min}}\rceil)$-approximation algorithm for continuous BGT.
\begin{enumerate}
\item[]{\hspace*{-\labelwidth}\hspace*{-\labelsep}[\emph{Pre-processing}]}
\item
Let $s= \lfloor\log_2({h_{\max}}/{h_{\min}})\rfloor+1$.
\item
For $i \in \{1,2\ldots,s\}$,
let $V_i= \{v_j\in V \mid 2^{i-1} \cdot h_{\min} \le h_j <  2^i \cdot h_{\min}\}$.
\item
{\bf for} $i=1$ {\bf to} $s$ such that $V_i \neq\emptyset$ {\bf do}
\item\hspace*{5mm} 
Let $T_i$ be a minimum spanning tree on $V_i$.
\item\hspace*{5mm}
Let
$C_i$ be a directed Euler-tour traversal of~$T_i$.
\item\hspace*{5mm}
Set an arbitrary point on $C_i$ as the \emph{last visited point} on $C_i$.
\item[]{\hspace*{-\labelwidth}\hspace*{-\labelsep}[\emph{Walking}]}
\item
{\bf repeat forever}
\item
\hspace*{5mm}
{\bf for} $i=1$ {\bf to} $s$ such that $V_i \neq\emptyset$ {\bf do}
\item
\hspace*{10mm}
Walk to the last visited point on $C_{i}$.
\item
\hspace*{10mm}
If $V_i$ has at least $2$ points, then walk along $C_i$ in the direction of the tour, \\
\hspace*{10mm}
stopping as soon when
the distance covered is at least $D$.
\end{enumerate}
\hrule
\vspace*{1mm}
\end{figure}

\begin{theorem}
Algorithm~2 is an $O(\log{\lceil h_{\max}}/{h_{\min}\rceil})$-approximation algorithm
for the continuous BGT problem.
\end{theorem}

\begin{proof}
For any $i \in \{1,2,\ldots,s\}$, consider any point $v \in V_i$.
The robot walks along one tree for  at most distance $2D$ and then covers at most distance $D$ to move to the next tree.
After a visit to point $v$, the robot comes back to tree $T_i$ at most 
$\left\lceil{2\cdot\MST(V_i)}/{D}\right\rceil$ times before visiting $v$ again.
Therefore, recalling from the algorithm that there are at most 
$s = \lfloor\log_2({h_{\max}}/{h_{\min}})\rfloor+1$ trees,
the distance traveled between two consecutive visits to point $v$ is at most
\begin{eqnarray*} 
3Ds \left\lceil\frac{2\cdot\MST(V_i)}{D}\right\rceil & \le & 3s (D + 2\cdot\MST(V_i)).
%
\end{eqnarray*}
Hence, the height of the bamboo at $v$ is always 
\begin{equation}
\label{eq:algo2a}
O\left(h_{\max}(V_i) \cdot \log\left\lceil\frac{h_{\max}}{h_{\min}}\right\rceil \cdot \max\{D,\MST(V_i)\}\right).
\end{equation}
On the other hand, using Lemmas~\ref{LB1} and \ref{LB2}, we obtain
\begin{equation}
\label{eq:algo2b}
\OPT \; \ge \; h_{\min}(V_i) \cdot \max\{D,\MST(V_i)\}.
\end{equation}
Combining the two bounds (\ref{eq:algo2a}) and (\ref{eq:algo2b}), and observing that 
$h_{\max}(V_i) \leq 2\cdot h_{\min}(V_i)$,
we obtain that the height of the bamboo at $v$ 
is always $O(\OPT \cdot \log\lceil {h_{\max}}/{h_{\min}}\rceil)$, so
Algorithm~2 is an $O(\log\lceil{h_{\max}}/{h_{\min}}\rceil)$-approximation algorithm for BGT.
\ep\end{proof}

\begin{figure}[t]
\vspace{-.2cm}
\hrule\vspace{.1cm}
\vspace{3mm}
\textbf{Algorithm 3:} An $O(\log n)$-approximation algorithm for continuous BGT.
\begin{enumerate}
\item[]{\hspace*{-\labelwidth}\hspace*{-\labelsep}[\emph{Pre-processing}]}
\item
Let $s= \lceil 2 \cdot \log_2 n\rceil$.
\item
Let $V_0= \{v_i\in V \mid h_i \leq n^{-2}\}$.~ Let $V_0=\{v'_0, v'_1, \ldots, v'_{\ell-1}\}.$
\item
For $i \in \{1,2,\ldots,s\}$, let $V_i= \{v_j\in V \mid 2^{i-1} \cdot n^{-2} < h_j \leq  2^i \cdot n^{-2}\}$.
\item
{\bf for} $i=1$ {\bf to} $s$ such that $V_i\neq\emptyset$ {\bf do}
\item\hspace*{5mm} 
Let $T_i$ be a minimum spanning tree on $V_i$.
\item\hspace*{5mm}
Let
$C_i$ be a directed Euler-tour traversal of~$T_i$.
\item\hspace*{5mm}
Set an arbitrary point on $C_i$ as the \emph{last visited point} on $C_i$.
\item[]{\hspace*{-\labelwidth}\hspace*{-\labelsep}[\emph{Walking}]}
\item
$j=0$.
\item
{\bf repeat forever}
\item
\hspace*{5mm}
{\bf for} $i=1$ {\bf to} $s$ such that $V_i\neq\emptyset$ {\bf do}
\item
\hspace*{10mm}
Walk to the last visited point on $C_{i}$.
\item
\hspace*{10mm}
If $V_i$ has at least $2$ points, then walk along $C_i$ in the direction of the tour, \\
\hspace*{10mm}
stopping as soon when
the distance covered is at least $D$.
\item
\hspace*{5mm}
If $V_0 \neq \emptyset$, then walk to $v'_{j}$ and let 
$j = j+1 \pmod \ell$.
\end{enumerate}
\hrule
\vspace*{1mm}
\end{figure}

\begin{theorem}
Algorithm~3 is an $O(\log n)$-approximation algorithm for the continuous BGT problem.
\end{theorem}

\begin{proof}
Each round of Algorithm~3, that is, each iteration of the repeat loop, is a cycle over all $s = \Theta(\log n)$ trees 
and a visit to one point in the set $V_0$.
Consider any point $v \in V_i$, for any $i \in \{1,2,\ldots,s\}$.
The distance traveled between two consecutive visits of $v$ is at most
\begin{eqnarray*}
(3Ds + D)\left\lceil\frac{2\cdot\MST(V_i)}{D}\right\rceil & \le & (3s+1) (D + 2\cdot\MST(V_i)) \\
& = & O(\log n \cdot \max\{D,\MST(V_i)\}).
\end{eqnarray*}
Hence, the height of the bamboo at $v$ is always
\begin{equation}
\label{eq:algo3a}
O(h_{\max}(V_i) \cdot \log n \cdot \max\{D,\MST(V_i)\}). \nonumber
\end{equation}
Using the lower bound~\eqref{eq:algo2b} on $\OPT$
and the fact that  $h_{\max}(V_i) \le 2 \cdot h_{\min}(V_i)$, we conclude that 
the height of the bamboo at $v$ is always $O(\OPT\cdot\log n)$.

Consider now a point $v \in V_0$.
The distance traveled between two consecutive visits of $v$ is at most
$$(3Ds + D)|V_0| \; = \; O(n \cdot D \cdot \log n).$$
Hence, the height of the bamboo at $v$ is always
\begin{equation}
\label{eq:algo3c}
O(h_{\max}(V_0) \cdot n \cdot D \cdot \log n)
= O(n^{-2} \cdot n \cdot D \cdot \log n)
= O( h_{\max} \cdot D \cdot \log n). \nonumber
\end{equation}
This and Lemma~\ref{LB1} imply that the height of the bamboo at $v$ is always 
$O(\OPT \cdot \log n)$.
Thus Algorithm~3 is an $O(\log n)$-approximation algorithm for BGT.
\ep\end{proof}

Note that the pre-processing in all Algorithms~1, 2 and 3 is dominated by the minimum-spanning tree computation,
which can be implemented in $O(n^2)$ time (e.g.~by using Prim's algorithm \cite{Prim57}). 
Then the running time to produce the schedule (the consecutive steps of the schedule) is constant per one step of the schedule.

\subsection{How tight are the upper and lower bounds?}

Our $O(\log n)$-approximation algorithm for the continuous BGT (Algorithm~3)
can return schedules which are worse than the optimal schedules by 
a factor of $\Theta(\log n)$. 
For example, consider the input which consists of two sets $V'$ and $V''$ of $n/2$ points each such that 
in each set the points are very close to each other (with the total distance to visit all points in this set  less than $D$),
but the sets are at distance greater than $D/2$ from each other. 
The rates of growth in set $V'$ include the $\Theta(\log n)$ values $1/4, 1/8, \ldots, 1/2^i, \ldots, 1/n$,
and the same rates are in set $V''$.
For this input instance the value of the optimal schedule is $\Theta(D)$: visit all points in $V'$, then all points in $V''$
(for the total distance $\Theta(D)$),
and repeat.
The schedule computed by Algorithm~3 uses $\Theta(\log n)$ trees and makes the robot traverse each tree for 
a distance at least $D$ before returning to the bamboo with the highest rate of growth of $1/4$. 
Thus this bamboo
grows to the height $\Theta(D\log n)$, which is a factor of $\Theta(\log n)$ worse than the optimum.

The approximation bounds which we presented 
in Section~\ref{sec:ContinuousAlgo} are derived 
by comparing the upper bounds on the maximum bamboo heights guaranteed by the algorithms
with the lower bounds shown in Section~\ref{sec:Continuous-LB}.
We show now a class of instances, for which
any schedule leads to bamboo heights greater than our lower bounds
by a $\Theta(\log n)$  factor.
Thus for these input instances our general lower bounds turn out to be weak,
while our $O(\log n)$-approximation algorithm computes in fact constant approximation
schedules.
To improve the approximation ratio of algorithms for the continuous BGT to $o(\log n)$,
one will therefore require both:
new algorithmic techniques as well as stronger lower bounds. 

We consider the following input $V^*$ for the continuous BGT problem.
The $n$ points in $V^*$ are placed on the Euclidean plane within a unit-radius circle.
The points lie evenly spaced along the spiral inside this circle, which starts at a point $A$ at distance $1/2$
from the center of the circle and swirls outward creating rings separated by distance $d_2 = n^{-1/3}$;
see Figure~\ref{fig:spiral}, but note that the drawing is not to scale.
We view the points in $V^*$ as ordered along the spiral, with the first point at $A$
and the Euclidean distance between two consecutive points
equal to $d_1 = n^{-2/3}$.
Thus the length of the spiral between two 
consecutive points in $V^*$ is equal to
$n^{-2/3}(1+o(1))$, so
the length of the part of the spiral which is 
occupied by the points in $V^*$ is equal to 
$n^{1/3}(1+o(1))$.
As the length of the spiral within the circle is at least $(1/2)\pi n^{1/3}$
(since $(1/2)/d_2 = n^{1/3}/2$ rings, each of them of length at least $\pi$), for any sufficiently large $n$, all points in $V^*$
are indeed inside the circle.
On the other hand, for a sufficiently large $n$, there
are two points in $V^*$ at distance 
at least $1$ from each other, so 
the diameter $D$ of $V^*$ satisfies 
$1 \le D \le 2$.

\begin{figure}[t]
\begin{center}
\includegraphics[scale=0.8]{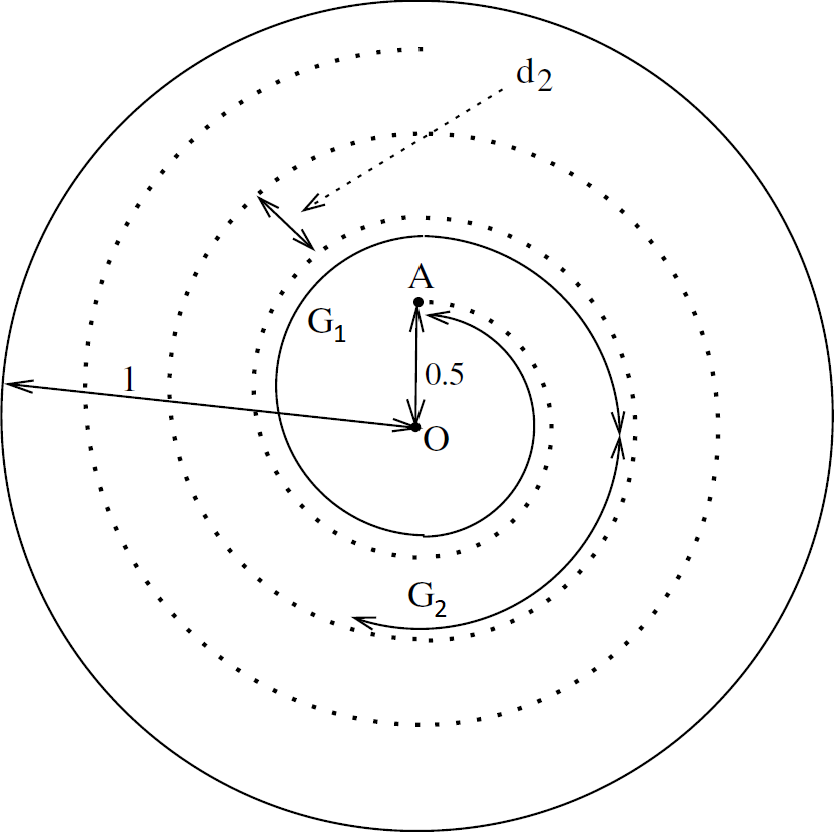}
\end{center}
\caption{Example of a spiral input.
\label{fig:spiral}}
\end{figure}

The points in $V^*$ are grouped into subsets $G_1, G_2, \ldots, G_{(\log n)/3}$ and $G'$
of consecutive points along the spiral, starting from position $A$.
Here $\log n = \log_2(n)$, and
for convenience we assume that $(\log n)/3$ is an integer.
The first two groups $G_1$ and $G_2$ are indicated in Figure~\ref{fig:spiral}.
For $i =1, 2, \ldots, (\log n)/3$, the size of group $G_i$ is
$n/2^{i}$
and each point in $G_i$ has the same growth rate $h_i = (3 -\epsilon) 2^i/(n\log n)$,
where $0 < \epsilon = o(1)$.
The last group $G'$ contains the remaining $o(n)$ points in $V^*$ and the growth rate of each point in $G'$
is equal to $h' = 1/n^{4/3} = h_{\min}(V^*)$.
The exact value of $\epsilon = o(1)$ is such that all growth rates sum up to $1$.

Since 
$h_{\max}(V^*) = h_{(\log n)/3} = \Theta(1/(n^{2/3}\log n))$,
Lemma~\ref{LB1} gives the lower bound of $\Omega(1/(n^{2/3}\log n))$ on the optimal solution.
We check now what lower bounds we can get from Lemma~\ref{LB2}.
The MST of this set of points $V^*$ is obtained by following the spiral and the weight of this MST is
equal to $nd_1 = n^{1/3}$, giving the lower bound of $h_{\min}(V^*) \cdot \MST(V^*) = 1/n$.
For each $i = 1,2, \ldots, (\log n)/3$, 
the weight of the MST of the set of points 
$V^{(i)} = \bigcup_{j = i}^{(\log n)/3} G_j$,
which is the subset of all points in $V^*$ with growth rates at least $h_i$,
is equal to
$d_1\left|\bigcup_{j = i}^{(\log n)/3} G_j\right| = \Theta(d_1 |G_i|) = \Theta(n^{1/3}/2^i)$.
This gives the lower bound of 
$h_i \cdot \MST(V^{(i)}) = \Omega(1/(n^{2/3}\log n))$.
This is the best lower bound which we can obtain from Lemma~\ref{LB2}, 
since $h_{\min}(V') \cdot \MST(V')$ 
is maximized by including in $V'$ \emph{all} 
points in $V^*$ with growth rates at least $h_{\min}(V')$.
Thus Lemmas~\ref{LB1} and~\ref{LB2} give for this input instance the lower bound
$\Omega(1/(n^{2/3}\log n))$.

The $O(\log n)$-approximation Algorithm 3 gives the schedule for $V^*$ with the maximum
bamboo height $\Theta(1/n^{2/3})$, which is a $\log n$ factor above our lower bounds.
Indeed, observe that for some $j\ge 0$,
each set $G_i$, 
$1 \le i \ge (\log n)/3$, is 
the sets $V_{j+i}$ in Algorithm~3.
Thus, to service all points in $G_i$, 
the algorithm needs
$\Theta(\MST(G_i)/D)$ iterations of the 
walking loop (the ``repeat forever'' loop).
The walking in each iteration takes 
$\Theta(D)$ time per each set $G_x$, so 
$\Theta(D\log n)$ time in total.
This means that 
a bamboo in set $G_i$ grows up to the height of 
$\Theta(h_i (D \log n) (\MST(G_i)/D)) = \Theta(1/n^{2/3})$.
We show next that for this input any possible schedule produces bamboos of height $\Omega(1/n^{2/3})$.

\begin{lemma}
For each schedule for the input $V^*$, there must be a bamboo which grows to the
height $d_1/2 = \Omega(1/n^{2/3})$.
\end{lemma}

\begin{proof}
Assume that there is a schedule
such that the height of each bamboo in $\bigcup_{i=1}^{(\log n)/3} G_i$ is always at most
$d_1/2$.
In such a schedule, each point in each set $G_i$ is serviced after at most distance
$(d_1/2) \cdot (n\log n)/((3-\epsilon)2^i)\le d_1 \cdot (n\log n)/2^{i+2}$.
Since the distance between each two points in $V^*$ is at least $d_1$, the growth rate of
each point in $G_i$ in this schedule must be at least  $2^{i+2} /(n \log n)$.
Thus the sum of the growth rates of all points in $\bigcup G_i$ in this schedule is at least
\[ \sum_{i = 1}^{(\log n)/3} \frac{2^{i+2}}{n\log n} \cdot \frac{n}{2^i} = 4/3,
\]
which is a contradiction since the growth rates of the points
in any valid schedule must sum up to at most $1$.
This contradiction implies that in each valid schedule there is a bamboo which
grows higher  than $d_1/2$.
\ep\end{proof}

\section{Conclusions and open questions}
\label{sec:conclusion}
There are several interesting open questions about approximation algorithms for the BGT problems.
For discrete BGT and simple strategies such as $\ReduceMax$ and $\ReduceFastest$, 
we still do not know the
exact upper bounds on the maximum heights of the bamboos, in relation to $H$, or the exact worst-case approximation ratios
(these two parameters are related but not 
the same).
%

Assuming the growth rates are 
normalized to $H=1$,
the best known upper bound 
on the maximum height of a bamboo
under the $\ReduceFastest$ strategy
is $2.62$ (Bil{\`{o}} {\em et al.}~\cite{BGLPS22}) and 
the best known lower bound 
is $2.01$ (Kuszmaul~\cite{DBLP:conf/spaa/Kuszmaul22}).
The first 
constant bound on the maximum height under the $\ReduceMax$ strategy, shown in~\cite{BGLPS22}, was $9$, and the
current best bound, shown 
in~\cite{DBLP:conf/spaa/Kuszmaul22}, is $4$,
while a simple example shows that bamboos can reach heights arbitrarily close to $2$. 
Can we decrease the upper bounds on the 
maximum height of bamboos in $\ReduceMax$ or
$\ReduceFastest$, or find examples to increase 
the lower bounds?
Similarly, there are gaps between the upper
and lower bounds on the approximation ratios
of $\ReduceMax$ and $\ReduceFastest$, where the upper bounds, so far, come only as straight consequences of the upper bounds 
on the maximum heights mentioned above, and 
the lower bounds are as discussed in 
Section~\ref{sec:Approx-simpleStartegy}.
%
%

Other sets of questions are about general bounds on approximation ratios for the BGT problem.
We showed in Section~\ref{sec:impoved-approx}
a $8/5+o(1)$-approximation algorithm, improving on the previous-best $12/7$-approximation by Van Ee~\cite{Ee21}. 
Can this bound be further improved, ideally achieving a PTAS algorithm?
Another question is whether the $1+O(\sqrt{h_1/H})$ approximation ratio of our algorithm for discrete BGT 
presented in Section~\ref{sec:DiscreteBGT-offline}
can be improved. Can we achieve an approximation ratio of $1+O({h_1/H})$, or even just $1+o({\sqrt{h_1/H}})$?

For continuous BGT, we do not know whether our Algorithm 3, or any other algorithm, achieves
an approximation ratio of  $o(\log n)$. 
The two special cases of continuous BGT -- discrete BGT and the metric TSP -- have both constant-ratio
polynomial-time approximation algorithms, not giving any indication why we should not expect the same for the 
more general problem. 
Note that a constant approximation for the path was given 
by Chuangpishit~{\em et al.}~\cite{DBLP:conf/sofsem/ChuangpishitCGG18},
and a PTAS was presented by Damaschke~\cite{DBLP:conf/iwoca/Damaschke20}.

In this paper we considered only the case of
one robot.
Damaschke~\cite{DBLP:conf/iwoca/Damaschke20,DBLP:conf/atmos/Damaschke21}
studies the case with two robots patrolling on a line.
The work of Bender~{\it et al.}~\cite{DBLP:conf/stoc/BenderFK19}
and Kuszmaul~\cite{DBLP:conf/spaa/Kuszmaul22}
on cup emptying games 
includes a methodology 
of transferring results from single- to multi-processor
scenarios.
Generally, however, the studies of perpetual scheduling for multi-robot scenarios have been so far
rather limited, but we expect that this will change.
%

\bibliographystyle{plain}
\bibliography{BGT_refs}

\end{document}